\documentclass{amsart}
\usepackage{enumitem}
\usepackage[all]{xy}
\usepackage[margin=1.00in, marginparwidth=1.5cm, marginparsep=0.5cm]{geometry}

\usepackage{lipsum} 

\usepackage{verbatim}

\usepackage{mathtools}

\usepackage{bbold}

\usepackage{tikz}
\usepackage{amssymb}
\usepackage{array}
\usepackage{xcolor}

\newcolumntype{C}[1]{>{\centering\arraybackslash}p{#1}}

\usepackage{mathrsfs}
\usepackage{physics}
\usepackage{amsthm}
\usepackage{amsmath,amsfonts,graphicx,mathdots,cite,wrapfig}
\usepackage[colorlinks,linkcolor=blue,citecolor=blue,pagebackref,hypertexnames=false, breaklinks]{hyperref}
\usepackage{cite}
\usepackage{mwe}
\usepackage{caption}
\usepackage{marginnote}

\usetikzlibrary{backgrounds}

\newtheorem{theorem}{Theorem}[section]
\newtheorem{lemma}[theorem]{Lemma}
\newtheorem{proposition}[theorem]{Proposition}
\newtheorem{corollary}[theorem]{Corollary}
\newtheorem{assumption}[theorem]{Assumption}

\theoremstyle{definition}
\newtheorem{definition}[theorem]{Definition}
\newtheorem{example}[theorem]{Example}
\theoremstyle{remark}
\newtheorem{remark}[theorem]{Remark}

\numberwithin{equation}{section}

\DeclareMathOperator{\integers}{{\mathbb{Z}}}
\DeclareMathOperator{\nn}{\mathbb{N}}
\DeclareMathOperator{\rr}{\mathbb{R}}
\DeclareMathOperator{\complex}{\mathbb{C}}
\DeclareMathOperator{\hamilton}{\mathbb{{\textbf{H}}}}

\DeclareMathOperator{\Lhamilton}{\mathbb{{\textbf{L}}}}
\DeclareMathOperator{\LhamiltonZ}{\mathbb{\textbf{L}}_{\integers}}
\DeclareMathOperator{\jacobi}{\textbf{J}}
\DeclareMathOperator{\Plax}{\textbf{P}}
\DeclareMathOperator{\PlaxZ}{\textbf{P}_{\integers}}
\DeclareMathOperator{\PlaxG}{\textbf{P}_{G}}
\DeclareMathOperator{\Id}{\textbf{Id}}
\DeclareMathOperator{\proj }{\textbf{Proj}}
\DeclareMathOperator{\pp}{{\widetilde{\pii}}}

\newcommand{\pii}{{\ensuremath{\mbox{\small$\Pi$\hskip.012em\llap{$\Pi$}\hskip.012em\llap{$\Pi$}\hskip.012em\llap{$\Pi$}}}}}
\renewcommand{\pii}{{\ensuremath{\mbox{$\Pi$\hskip.012em\llap{$\Pi$}\hskip.012em\llap{$\Pi$}\hskip.012em\llap{$\Pi$}}}}}

\begin{document}


\title{Hamiltonian systems, Toda lattices, Solitons, Lax Pairs on weighted $\integers$-graded graphs}


\author{Gamal Mograby}
\address{Gamal Mograby, Mathematics Department, University of Connecticut, Storrs, CT 06269, USA}
\email{gamal.mograby@uconn.edu}

\author{Maxim Derevyagin}
\address{Maxim Derevyagin, Mathematics Department, University of Connecticut, Storrs, CT 06269, USA}
\email{maksym.derevyagin@uconn.edu}

\author{Gerald V. Dunne}
\address{Gerald V. Dunne, Mathematics \& Physics Department, University of Connecticut, Storrs, CT 06269, USA}
\email{gerald.dunne@uconn.edu}

\author{Alexander Teplyaev}
\address{Alexander Teplyaev, Mathematics \& Physics Department, University of Connecticut, Storrs, CT 06269, USA}
\email{alexander.teplyaev@uconn.edu}


\date{\today}


\begin{abstract}
{We consider discrete one dimensional nonlinear equations and present the procedure of lifting them to} 
$\integers$-graded graphs. We identify conditions which allow one 
to lift one dimensional solutions  to solutions on graphs. In particular, we prove the existence of  solitons {for static potentials} on graded fractal graphs.  
{We also show that even for  a simple example of a topologically interesting  graph the corresponding non-trivial Lax pairs and associated unitary transformations 
 do not lift to a Lax pair on the $\integers$-graded graph. }
\tableofcontents\end{abstract}
\maketitle

\section{Introduction}

In this paper we extend to the realm of {\it nonlinear} differential equations our recent investigations of the rich interplay between Brownian motion, Laplacians, energy forms and geometries beyond one-dimensional graphs. Even for linear problems, there are remarkable connections between fractal-like geometries and spectral problems, and here we begin to explore similar mathematical relations for nonlinear systems. The natural place to start is with the theory of {\it integrable} nonlinear equations, for which the Toda lattice is  { one of the well studied important cases.} Integrable equations are of great interest in mathematics and in physical applications 
\cite{shabat1972exact,ablowitz1991solitons,agrawal2013nonlinear,NMPZ,FaddeevTakhtajan}.

For one-dimensional integrable models, integrability may be characterized in several  inter-related ways, and here we adopt the Lax pair and Hamiltonian system formulations. An important technical challenge is to define a fractal first derivative operator, which naturally appears in Lax pairs for one-dimensional systems.
We analyze the lifting (in the sense defined below) of one-dimensional integrable non-linear differential equations to  $\integers$-graded graphs with fractal properties. We prove that this lifting procedure leads to the existence of static soliton solutions, but show that there is a natural obstruction to the preservation of the underlying Lax pair structure, except in a projected subspace.  

 In \cite{2019arXiv190908668D,mograby2020spectra} we showed that a certain class of fractal-type graphs demonstrates favorable geometrical properties when used as ambient spaces for quantum systems. 
 {Briefly speaking, we considered a transversal layer structure of finite graphs (for more details, see \cite[Definition 2.1]{mograby2020spectra}). In fact, similar concepts can be found in \cite{Fomin,Stanley} under the name ``graded graphs'' or in \cite[page 76]{MR2316893} as a ``stratification'' of a graph}.  {Here we apply these constructions to fractal-like graphs.} 
We interpret the graphs in question as qubit networks, in which each vertex represents a spin. The spin-spin interaction is described by a given Hamiltonian operator. The transversal decomposition of such a graph provides an auxiliary  $1D$ chain. We interpret this $1D$ chain in a similar manner as a $1D$ qubit chain, where the spin-spin interaction is described by a given Jacobi matrix.
In  \cite{2019arXiv190908668D,mograby2020spectra} we introduced a technique on how to relate the qubit network on the graph and the auxiliary $1D$ chain. By doing so, we were able to reduce some quantum systems on graphs to one-dimensional solvable models. 
These methods have been applied successfully to analyze perfect quantum state transfer on graphs. In particular, we gave an algorithm to design Hamiltonians on a large class of graphs, such that perfect quantum state transfer is achieved. This result was  established by reducing the perfect quantum state transfer to the auxiliary $1D$ chain and subsequently lifting the known results of the one-dimensional case to the graph. 

\begin{figure}[th]
\centering
\includegraphics{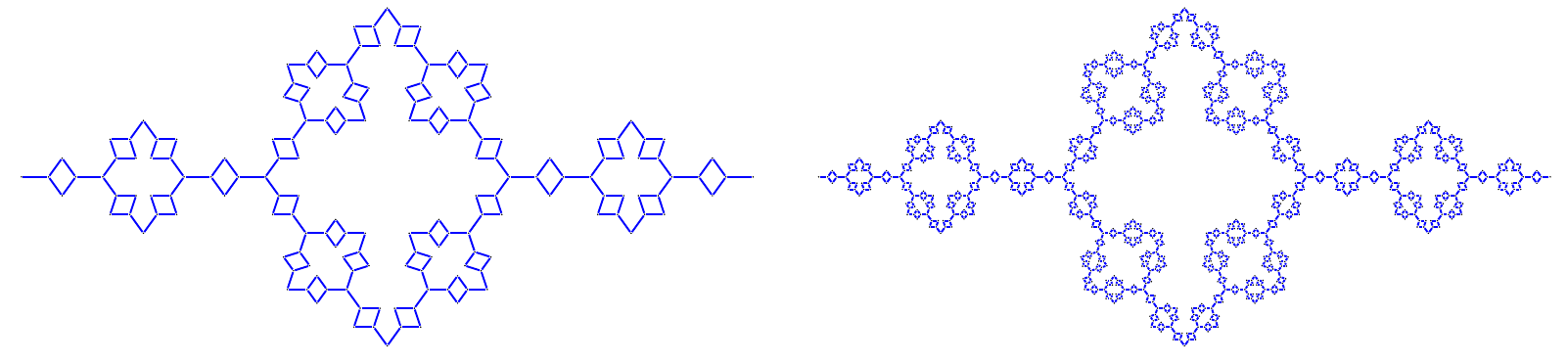}
	\caption{Hambly-Kumagai  Diamond finite   $\integers$-graded graphs   of level 3 and 4 with uniformly bounded degree  and the asymptotic similarity dimension $\text{dim}=\frac{\log6}{\log4}$, \cite{KT,MT,MT2,LP}.}
	\label{fig:Hambly-Kumagai}
\end{figure}
\begin{figure}[thb]
	\centering
	\includegraphics{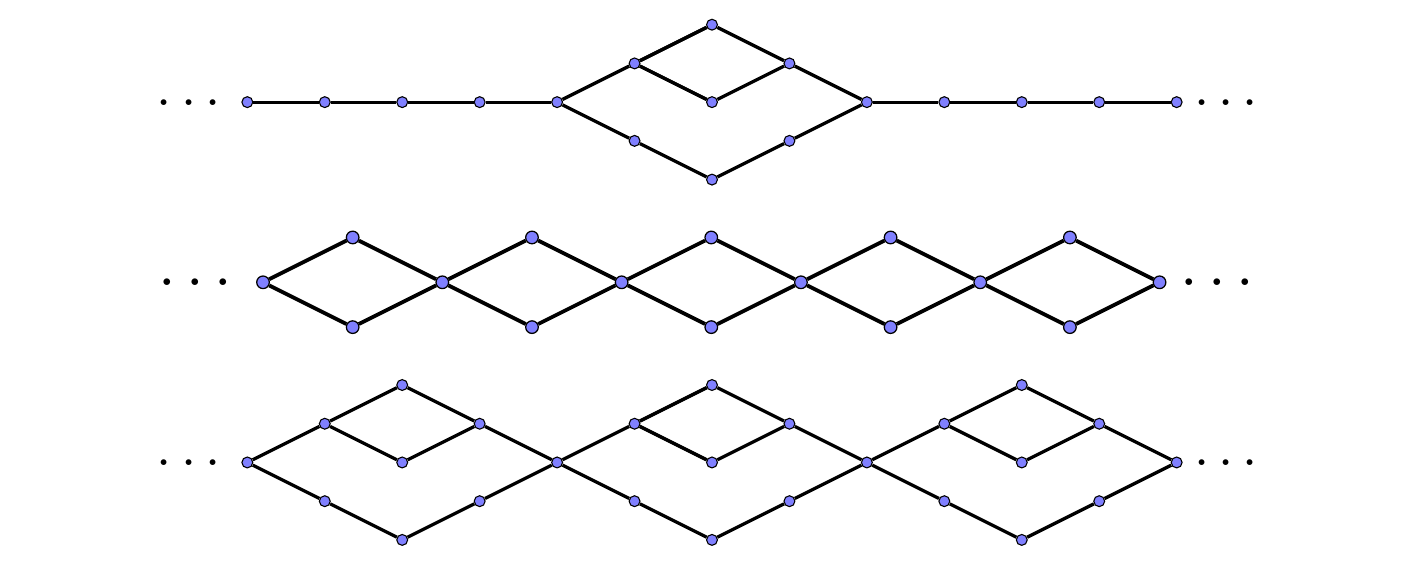}
	\caption{Three simple examples of infinite $\integers$-graded graphs where  our methods are applicable. }
	\label{fig:simple}
\end{figure}

In the present paper, we extend our approach to classical {\it nonlinear} Hamiltonian systems. As a particular example we consider the Toda lattice, which was first introduced in \cite{doi:10.1143/JPSJ.22.431,doi:10.1143/JPSJ.23.501} as a model of an infinite $1D$-chain of particles with nearest-neighbor interactions. It describes a simple one-dimensional crystal composed of particles that interact by nonlinear forces \cite{nla.cat-vn2968918}. The corresponding equations of motion are formulated as a classical Hamiltonian system, in which  $n \in \integers$ corresponds to the $n$-th particle. The Hamiltonian function in the case of exponential interactions has the following form
\begin{equation}
\label{eq:oneDHamilForToda}
H_{\integers}(q,p) = \sum_{n \in \integers}  \frac{p_{\integers}^2(n,t) }{2}   +   \sum_{(n,n+1 ) \in E(\integers)}  U( q_{\integers}(n+1,t)-q_{\integers}(n,t)).
\end{equation}
where $q(n,t)$ defines the displacement  at the moment $t$ of the $n$-th particle from its equilibrium position, $p(n,t)$ represents the momentum and $U(r) = e^{-r} -1$ is a potential function, which accounts for the exponential interactions. Note that the Hamiltonian function  defined in (\ref{eq:oneDHamilForToda}) is a rescaled variant of the original version introduced by M. Toda (for more details, see\cite[Chapter 12]{TeschlbookJacobi}). Toda's seminal papers in 1967 generated broad interest in both mathematics as well as physics communities. Flaschka \cite{PhysRevB.9.1924}, Henon \cite{PhysRevB.9.1921} and Moser \cite{Moser1975} proved the complete integrability of the Toda lattice (nonperiodic and periodic Toda lattices). Moreover, Flaschka \cite{PhysRevB.9.1924} and Manakov \cite{manakov} introduced a change of variables and thereby reformulated the equations of motions of a Toda lattice as a Lax Pair \cite{doi:10.1002/cpa.3160210503}.  This gives rise to a time-evolution on the vector space of real, symmetric, tridiagonal matrices which preserves the spectrum of the initial data. Another prominent feature of Toda lattices, which is connected to complete integrability, is the existence of soliton solutions. 

For graphs in our paper we understand the notion of solitons in a wide sense, i.e., as nonlinear waves that tend to constant values at infinity and have certain stability properties. Note that our graphs may not be translationally invariant, and therefore we are not defining solitons as translational invariant solutions. 
However, if we assume that the graph is translationally invariant with a period $T$, and a naturally defined radial soliton 
is moving with a speed $c$, then this soliton   keeps its profile after time shift by $T/c$ and the space shift by $T$ because of Theorem~\ref{thm:liftToradailSolution}. Two examples of the periodic graphs with periods $T=2$ and $T=4$ are the bottom two graphs in Figure~\ref{fig:simple}. It is not possible to get a generalization of Lax pairs  because of reasoning in Theorem~\ref{noLaxTheorem} (see also  Section~\ref{non-Abelian} about a relation to the non-Abelian Toda lattice).

We denote a graph by $G=(V(G), E(G))$, where $V(G)$ and $E(G)$ represent the set of vertices and directed edges, respectively. We equip the set of vertices and edges with measures
\begin{equation*}
  \mu_V: V(G) \to (0, \infty), \ x     \mapsto \mu_V(x) ,\quad \quad  \quad  \mu_E: E(G) \to (0, \infty), \  (x,y)  \mapsto \mu_E(x,y).
\end{equation*}
The potential energy term in the Hamiltonian function will be expressed in terms of the graph edges. In this way, the particle interactions are reflected in the graph adjacency relations. For a potential, we set a function $U:\rr \to \rr$, which is assumed to be twice continuously differentiable, i.e., $U \in C^2(\rr)$. A Hamiltonian function on a $\integers$-graded graph $G$ is then formally defined as
\begin{equation}
\label{hamiltonOnGraphForIntro}
H_G(q_t,p_t) = \sum_{x \in V(G)}  \frac{p^2(x,t) }{2\mu_V(x)}   +   \sum_{(x,y ) \in E(G)} \mu_E(x,y) U( \partial q_t(x,y)), 
\end{equation}
where $\partial q_t(x,y) := q(y,t) - q(x,t) , \ (x,y) \in E(G)$, see Definition \ref{def:Classical Hamiltonian}. {Note that $\integers$-graded graphs are defined at the beginning of section  \ref{sec-Z}.}
Having introduced a Hamiltonian function on a   graph, we then compute the corresponding equations of motion. Subsequently, we rewrite the equations of motion as a flow in an appropriately chosen Banach space. By doing so, and under the assumption that there exists $ \delta >0$ such that $\mu_V(x) \geq \delta $ for all $x \in V(G)$, we  prove local in time existence and uniqueness of solutions for the equations of motion in Theorem \ref{existanceUniqueTheorem}. We note that the considerations so far are true for general graphs as there was no particular use of the $\integers$-graded graph definition.
An essential aspect of $\integers$-graded graphs, see Figures~\ref{fig:Hambly-Kumagai} and \ref{fig:simple}, is that they admit a natural ``radial direction'', which gives rise to the transversal decomposition and the auxiliary $1D$ chain. We equip the $\integers$-graded graph and the associated $1D$ chain with the Hamiltonian functions (\ref{hamiltonOnGraphForIntro}) and (\ref{eq:oneDHamilForToda}), respectively.
One of our main results in section \ref{sec-CHS} is Theorem \ref{thm:liftToradailSolution}, which states that under the Assumption \ref{necessaryAssumptionforRadialSol} on the measures  we can lift any solution corresponding to the Hamiltonian function (\ref{eq:oneDHamilForToda}) to a radial solution corresponding to the Hamiltonian function  (\ref{hamiltonOnGraphForIntro})  in the sense of Definition~\ref{def:liftToRadialSol}.  

Theorem \ref{thm:liftToradailSolution} provides a way to find  soliton solutions on $\integers$-graded graphs. Note that this result is very general and does not make any explicit reference to the type of the potential $U$, except the regularity assumption, $U \in C^2(\rr)$. Another implication of Theorem \ref{thm:liftToradailSolution} is that the radial solutions are locally stable in time, i.e.,  if the initial value is taken in the subspace of radial functions, then the corresponding solutions stay in this subspace. The stability of radial solutions will be used in the last section when investigating the Lax pair formalism on $\integers$-graded graphs. In Section~\ref{sec-Toda}, we specialize to the case of Toda lattices. Using known results about the one-dimensional Toda lattice, Corollary \ref{globalSolutionsAndRadial} extends the result in Theorem \ref{thm:liftToradailSolution} and provides a global radial solution, i.e., it exists for all $ \ t \in \rr$.  In particular, an N-soliton solution  can be lifted with the same argument to a radial solution on a $\integers$-graded graph $G$. 

The last section, Section \ref{sec:LiftingLaxPair}, deals with the Lax pair formalism on $\integers$-graded graphs. The intuition behind our approach is the following. We transform the equations of motion for the $1D$ chain $\integers$ via the Flaschka’s variables, see Definition \ref{FlaschkaVariables}. We obtain an equivalent system of equations, which is a time evolution given by a Lax pair $\{\PlaxZ  (t), \LhamiltonZ(t)\}$, namely
\begin{equation}
\label{eq:firstLaxEqForIntr}
\frac{d}{dt} \LhamiltonZ(t) = [\PlaxZ  (t), \LhamiltonZ(t)]:=\PlaxZ  (t) \LhamiltonZ(t)- \LhamiltonZ(t)\PlaxZ  (t), \quad t \in \rr,
\end{equation}
where $ \LhamiltonZ(t)$ is an infinite Jacobi matrix and  $\PlaxZ  (t) := [\LhamiltonZ(t)]_{+} - [\LhamiltonZ(t)]_{-}$, such that  $[\LhamiltonZ(t)]_{+}$  and $ [\LhamiltonZ(t)]_{-}$ define the upper and lower triangular parts of $\LhamiltonZ(t)$ respectively.  In Theorem \ref{thm:LiftingTheorem} we proved that under some assumptions each Jacobi matrix could be lifted to an operator on a $\integers$-graded graph. We call such an operator a lifted Jacobi matrix and denote it by $\Lhamilton(t)$. We raise the question of whether we can find a skew-adjoint operator $\PlaxG (t)$ acting on $G$ such that 
\begin{equation}
\frac{d}{dt} \Lhamilton (t) = [\PlaxG (t) ,\  \Lhamilton(t) ].
\end{equation}
In such a  case  we can imply that the pair $\{\PlaxG (t), \Lhamilton(t)\}$ satisfies the isospectral property. We constructed an example showing that this is not always possible, see Theorem \ref{noLaxTheorem}.

Our work is part of a long term study of  mathematical physics on fractals and self-similar graphs 
\cite{v1,v2,ADT09,ADT10,ABDTV12,ACDRT,Akk,Dunne12,hanoi,HM19}, in which novel features of quantum processes on fractals can be associated with the unusual spectral and geometric properties of fractals compared to regular graphs and smooth manifolds.

The paper is organized as follows. 
In Section~\ref{sec-Z}   we develop  {the technique of lifting} operators from a $1D$ chain to $\integers$-graded graphs. 
Section~\ref{sec-CHS}   introduces classical Hamiltonian systems on  $\integers$-graded graphs and  provides existence and uniqueness statements for the equations of motion. 
In Section~\ref{sec-Toda}   some of the known results of the one-dimensional Toda lattice are lifted to $\integers$-graded graphs. 
In Section~\ref{sec:LiftingLaxPair}   we provide an example showing that a Lax pair formalism can not be lifted in general to a $\integers$-graded graph $G$.

\section{Lifting Operators to a $\integers$-Graded Graph}\label{sec-Z}
In this section, we mostly extend the machinery developed in previous work  \cite{2019arXiv190908668D,mograby2020spectra} to the case of infinite graphs.
We follow \cite{Fomin,Stanley,MR941434} and define a $\integers$-graded graph. We refer to the triple  $G = (V(G),E(G),  \pii )$ as a $\integers$-graded graph provided that
\begin{enumerate}
	\item  $(V(G),E(G))$ is a connected graph with a countable vertex set $V(G)$ and an edge set $E(G)$,
	\item $ \pii : V(G) \to \integers$ is a rank function,
	\item for an edge $(x,y) \in E(G)$, we have $ \pii (y)= \pii (x)+  1$. 
\end{enumerate}
This definition is illustrated in Figures~\ref{fig:Hambly-Kumagai}, \ref{fig:simple} and \ref{fig:simplestZgradedGraph}. 
\begin{figure}[htb]
	\centering
	\includegraphics{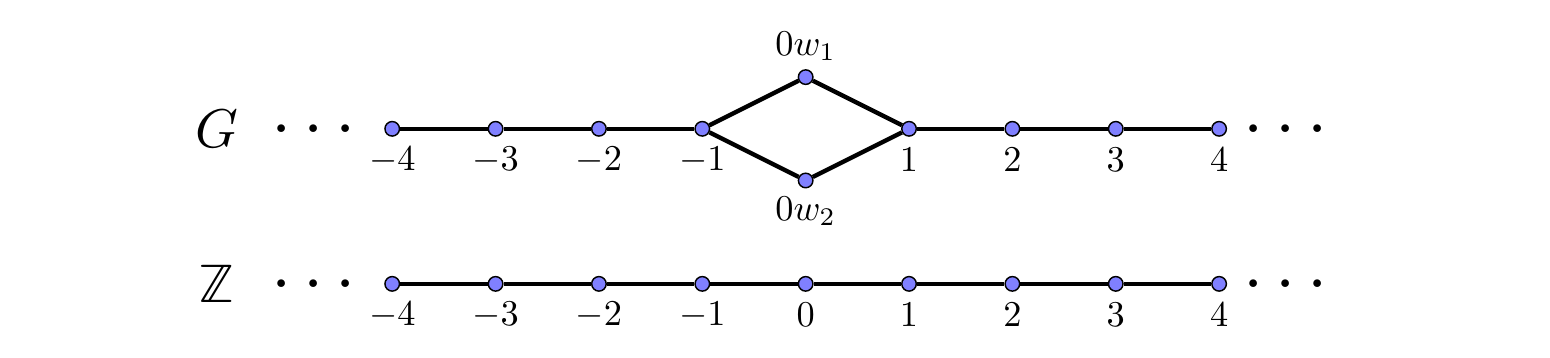}
	\caption{One of the simplest nontrivial infinite $\integers$-graded graphs $G$ with its projection on $ \pii (G)=\integers$.}
	\label{fig:simplestZgradedGraph}
\end{figure}

We call $  \pii ^{-1}(n) $ the $n$-th transversal layer. The graph $G$ is said to admit a countable transversal decomposition, $V(G) = \bigcup_{n \in \integers }   \pii ^{-1}(n) $.
\begin{remark}
The  rank function is commonly denoted by $\rho$, see \cite
{MR941434,
	Fomin,Stanley}. Similarly, the $n$-th transversal layer $  \pii ^{-1}(n) $ is often called the $n$-th  level of G  (see \cite[page 86]{borodin_olshanski_2016}). 
\end{remark}
Given $G = (V(G),E(G),  \pii )$ a $\integers$-graded graph. We equip $G$ with a measure, $\mu_V: V(G) \to [0, \infty]$. 
\begin{assumption}
\label{def:measure}
We assume $\mu_V(x)>0$ for all $x \in V(G)$ and that the restriction of $\mu_V$ to a transversal layer is a probability measure, i.e , $\mu(  \pii ^{-1}(n))=1$, $\forall \  n \in \integers$.
\end{assumption}
We showed in \cite{2019arXiv190908668D,mograby2020spectra} that a transversal decomposition $ V(G) = \bigcup_{n}   \pii ^{-1}(n) $ relates the graph $G$  in a natural way to an auxiliary 1D chain. The 1D chain in the case of a $\integers$-graded graph is simply the set of integers $\integers$ considered as a path graph  $(V(\integers), E(\integers))$ with the vertex set $V(\integers) = \integers$  and an edge set $E(\integers) =\{(n-1,n) : \ n \in \integers \}$. The graph $G$ and the associated $1D$ chain $\integers$ are equipped with the following Hilbert spaces.

\begin{eqnarray}
\label{key-inner}
 \ell^2(G)=\{\psi:V(G) \to \complex \ | \ \bra{\psi}\ket{\psi}_G  < \infty \}, \quad	& \quad  \bra{\psi}\ket{\varphi}_G = \sum_{x \in V(G)} \psi(x) \overline{\varphi(x)} \mu_V(x), \\
 \ell^2(\integers)=\{ \psi: \integers \to \complex | \bra{\psi}\ket{\psi}  < \infty \},\quad & \bra{\psi}\ket{\varphi} = \sum_{n \in \integers} \psi(n) \overline{\varphi(n)}.
\end{eqnarray}

Similar to the set of vertices, the $\integers$-graded graph structure induces a natural transversal decomposition of the set of edges. We equip the set of edges with a measure $\mu_E: E(G) \to [0, \infty], \  (x,y)  \mapsto \mu_E(x,y)$. 
\begin{assumption}
\label{def:EdgesMeasure}
We assume $\mu_E(x,y)>0$ for all $(x,y) \in E(G)$ and that the restriction of $\mu_E$ to a transversal layer of edges is a probability measure.
\end{assumption}
We will consider functions defined on edges and therefore introduce the following Hilbert spaces.
\begin{eqnarray*}
 \ell^2(E(G))=\{\psi:E(G) \to \complex \ | \ \bra{\psi}\ket{\psi}_E  < \infty \}, & \   \bra{\psi}\ket{\varphi}_E = \sum_{(x,y) \in E(G)} \psi(x,y) \overline{\varphi(x,y)} \mu_E(x,y), \\
 \ell^2(E(\integers))=\{ \psi: E(\integers) \to \complex | \bra{\psi}\ket{\psi}  < \infty \},\quad & \bra{\psi}\ket{\varphi} = \sum_{(n,n+1) \in E(\integers)} \psi(n,n+1) \overline{\varphi(n,n+1)}.
\end{eqnarray*}

Let $X$ be either  $V(G)$, $V(\integers)$, $E(G)$ or $E(\integers)$, we denote the Banach spaces of bounded functions on $X$ by

\begin{equation*}
 \ell^{\infty}(X)=\{\psi: X \to \complex \ | \    ||\psi||_{\ell^{\infty}(X)}  < \infty \}, \quad	||\psi||_{\ell^{\infty}(X)} := \sup_{x \in X} |\psi(x)|.
\end{equation*}
\begin{definition}
\label{def:radProjAvg}
Given $G = (V(G),E(G),  \pii )$ a $\integers$-graded graph.  The subspace of radial functions is defined as
$ \ell^2_{rad}(G)=\{\psi \in \ell^2(G) \ \ | \  \psi(x)=\psi(y) \text{ if }  \pii (x)= \pii (y) \}$. We denote the projection of $\ell^2(G)$ onto $\ell^2_{rad}(G)$ by $\proj  : \ell^2(G) \to \ell^2_{rad}(G)$. In addition, we define the averaging operator as the following mapping

\begin{equation*}
  \pp : \ell^2(G)  \to  \ell^2(\integers), \quad  \psi \mapsto    \pp \psi(n):=\sum_{x \in  \pii ^{-1}(n)}  \psi(x) \mu_V(x).
 \end{equation*}
\end{definition}
\begin{remark}
A simple computation shows that $\proj  $ is an orthogonal projection. 
\end{remark}
\begin{proposition}
\label{propForaveraging}
The averaging operator $ \pp $ is bounded with $|| \pp || = 1$. Let $ \pp ^{\ast} $ be the adjoint operator of $ \pp $, then $ \pp ^{\ast}$ is given by $ \pp ^{\ast}: \ell^2(\integers)  \to  \ell^2(G)$, $ \ \varphi  \mapsto  \pp ^{\ast} \varphi (x) = \varphi( \pii (x))$. Moreover, $ \ell^2_{rad}(G)$ is a closed subspace of $\ell^2(G)$.
\end{proposition}
\begin{proof}
Using $| \pp \psi(n)| \leq \sum_{x \in  \pii ^{-1}(n)}  |\psi(x)| \mu_V(x)$ and Jensen inequality, we imply
\begin{equation*}
| \pp \psi(n)|^2 \leq  \left( \sum_{x \in  \pii ^{-1}(n)}  |\psi(x)| \mu_V(x) \right)^2 \leq \sum_{x \in  \pii ^{-1}(n)}  |\psi(x)|^2 \mu_V(x).
\end{equation*}
Hence, $|| \pp \psi||^2 = \bra{ \pp \psi}\ket{ \pp \psi} = \sum_{n \in \integers} | \pp \psi(n)|^2 \leq  \sum_{n \in \integers}\sum_{x \in  \pii ^{-1}(n)}  |\psi(x)|^2 \mu_V(x) = \bra{\psi}\ket{\psi}_G = ||\psi||_G^2 $ and $|| \pp || \leq 1$. Equality holds by choosing $\psi$, for example, to be one on a single transversal layer and zero elsewhere. The second statement follows from \cite[Lemma 2.4,  page 3]{mograby2020spectra}. We prove now that $ \ell^2_{rad}(G)$ is a closed subspace of $\ell^2(G)$. Let $\psi \in \ell^2(G)$ and $\{\psi_m\}_{m \in \nn} \subset \ell^2_{rad}(G)$ such that $\lim_{m \to \infty}||\psi_m - \psi||_G=0$. By the nonnegativity of the terms, we have 

\begin{equation}
\label{eq:subspaceClosedness1}
\lim_{m \to \infty} \sum_{x \in  \pii ^{-1}(n)} |\psi_m(x) - \psi(x)|^2 \mu_V(x) = 0, \quad \forall n \in \integers.
\end{equation}
Recall Assumption \ref{def:measure}, we have $\mu_V(x) > 0$ for all $x \in V(G)$. Equation (\ref{eq:subspaceClosedness1}) implies for a fixed transversal layer $n \in \integers$, 
\begin{equation}
\lim_{m \to \infty}|\psi_m(x) - \psi(x)|= 0, \quad \forall x \in  \pii ^{-1}(n).
\end{equation}

Let $x, y \in  \pii ^{-1}(n)$, note that $\psi_m(x) = \psi_m(y)$ for all $m \in \nn$. We observe

\begin{equation}
\label{eq:subspaceClosedness2}
|\psi(x) - \psi(y)| \leq |\psi(x) - \psi_m(x)|+|\psi(y) - \psi_m(y)| ,\quad \forall m \in \nn.
\end{equation}
Taking the limit of inequality (\ref{eq:subspaceClosedness2}) gives $\psi(x) = \psi(y)$. Hence $\psi \in \ell^2_{rad}(G)$.
\end{proof}
We will use the following lemma frequently. 
\begin{lemma}
\label{usefulProp1}
Let $\Id_{\ell^2(\integers)}: \ell^2(\integers) \to \ell^2(\integers)$ be the identity operator on $\ell^2(\integers)$. Then
\begin{enumerate}
	\item The range of $ \pp ^{\ast}$ is $\ell^2_{rad}(G)$.
	\item $ Ker  \pp  = \ell^2_{rad}(G)^{\bot}$ and $ Ker  \pp ^{\bot} = \ell^2_{rad}(G)$
	\item $ \pp   \pp ^{\ast} =\Id_{\ell^2(\integers)}$
	\item $ \pp ^{\ast}  \pp  = \proj   $
\end{enumerate}
\end{lemma}
\begin{proof}
It follows by a similar argument for the case of finite graphs in \cite[Lemma 2.5,  page 4]{mograby2020spectra}. The second statement of part (2) holds by  $ Ker  \pp ^{\bot} = \overline{\ell^2_{rad}(G)}=\ell^2_{rad}(G)$ and Proposition \ref{propForaveraging}.
\end{proof}
The following considerations are relevant for lifting the Lax pair formalism from the auxiliary 1D chain to a $\integers$-graded graph (see section \ref{sec:LiftingLaxPair}). Let   $\{a(n,n+1) \}_{(n,n+1) \in E(\integers)} \in \ell^{\infty}(E(\integers))$  and $\{b(n) \}_{n) \in V(\integers)} \in \ell^{\infty}(E(\integers))$. We denote an \textit{infinite Jacobi matrix} by $\jacobi$
\begin{equation}
\label{eq:infiniteJacobiMatrix-}
\jacobi:    \ell^{2}(\integers) \to \ell^{2}(\integers), \quad   \quad 
	  \varphi(n)  \mapsto a(n,n+1)  \varphi(n+1) + b(n)  \varphi(n) + a(n-1,n)  \varphi(n-1),
\end{equation}
We require that both sequences are real-valued and $a(n,n+1) \neq 0$ for all $n \in \integers$. The Jacobi matrix $\jacobi$ is bounded and self-adjoint on $\ell^2(\integers)$.
\begin{assumption}[{Measure Balance Assumption}]
\label{necessaryAssumptionforRadialSol}
Let $x \in V(G)$. The following identities for the measures hold:
\begin{equation*}
\mu_V(x) =   \sum_{ (x,y ) \in E(G)} \mu_E(x,y), \quad \quad
\mu_V(x)  =  \sum_{(y,x ) \in E(G)} \mu_E(y,x).
\end{equation*}
\end{assumption}
\begin{remark}
The assumption \ref{necessaryAssumptionforRadialSol} is compatible with the assumptions \ref{def:measure} and \ref{def:EdgesMeasure}, i.e., 
\begin{equation*}
1 = \sum_{x \in  \pii ^{-1}(n)} \mu_V(x) =  \sum_{x \in  \pii ^{-1}(n)}  \sum_{ (x,y ) \in E(G)} \mu_E(x,y) = 1.
\end{equation*}
The assumption  \ref{necessaryAssumptionforRadialSol} arises naturally in the context of  
 \cite{ST19,AR18,AR21}.
\end{remark}
\begin{definition}
\label{def:LiftOperator}
Let $A: \ell^2(G) \to \ell^2(G)$. If $A$ reflects the adjacency relation of the graph $G$, in the sense that its off-diagonal elements are zero whenever they correspond to non-adjacent vertices, then we say \textit{$A$ acts on $G$}. Moreover, we define  
$B: \ell^2(\integers) \to \ell^2(\integers)$ such that $B =   \pp  A  \pp ^{\ast}$. In this case we say, the operator $A$ acts on $G$ and lifts $B$.  
\end{definition}
An operator $\hamilton$ acting on $\ell^2(G)$ is said to satisfy the \textit{Spectral Separation Assumption} if the following holds.
\begin{assumption}[{Spectral Separation Assumption}]
\label{separationAnsatz}
{A bounded linear} operator $\hamilton: \ell^2(G) \to \ell^2(G)$ satisfies 
\begin{enumerate}
	\item $\hamilton (Ker  \pp ) \subset Ker  \pp $, i.e.,   the subspace $Ker  \pp $ is invariant under $\hamilton$,
	\item $\hamilton (\ell^2_{rad}(G)) \subset  \ell^2_{rad}(G)$, i.e.,   the subspace $\ell^2_{rad}(G)$ is invariant under $\hamilton$.
\end{enumerate}
\end{assumption}
\begin{theorem}
\label{thm:LiftingTheorem}
Given $G = (V(G),E(G),  \pii )$ a $\integers$-graded graph. Under the Assumptions  \ref{def:measure}, \ref{def:EdgesMeasure} and \ref{necessaryAssumptionforRadialSol}   for each bounded infinite Jacobi matrix $\jacobi$  there exists an operator $\hamilton$ acting on $G$ in the sense of Definition \ref{def:LiftOperator} such that
\begin{equation}
\label{inducedJ}
\jacobi =   \pp  {\mbox{$\hamilton$}}  \pp ^{\ast}.
\end{equation}
{Moreover, $\hamilton$  can be directly computed with the formula  
\begin{equation} 
\label{eq:equForliftingThm}
 \hamilton(x,y) =
  \begin{cases}
     b(n)     & \quad \text{ if } x=y \text{ and }  \pii (x)=n, \\
    \frac{\mu_E(x,y)}{\mu_V(x)}  a(n,n+1) & \quad \text{ if }  (x, y) \in E(G) \text{ and }  \pii (x)=n, \\
    \frac{\mu_E(y,x)}{\mu_V(x)}  a(n-1,n) & \quad \text{ if }  (y, x) \in E(G) \text{ and }  \pii (y)=n-1 ,\\
    0    &    \quad \text{ otherwise. }
  \end{cases}
\end{equation}
$\hamilton$ defined by (\ref{eq:equForliftingThm}) satisfies the Spectral Separation Assumption \ref{separationAnsatz}.}
\end{theorem}
\begin{proof}
{We define $\hamilton = [\hamilton(x,y)]_{x,y \in V(G)}$ with the matrix elements given in (\ref{eq:equForliftingThm}).} The matrix product $\hamilton  \pp ^{\ast}$  has an intuitive meaning. Namely, it adds all columns corresponding to vertices in the same transversal layer to a single column. For example, let $x \in  \pii ^{-1}(n)$. Recall that the adjacent vertices of $x$ are in $ \pii ^{-1}(n+1)$. In other words, the non-zero off-diagonal elements in the row corresponding to the vertex $x$ are of the form $\frac{\mu_E(x,y)}{\mu_V(x)}  a(n,n+1)$, where $y \in  \pii ^{-1}(n+1)$. We obtain, for the sum

\begin{equation}
\label{measureConditionForProof}
  \sum_{ (x,y ) \in E(G)} \frac{\mu_E(x,y)}{\mu_V(x)}a(n,n+1) =  a(n,n+1),
\end{equation}
where the equality holds, due to assumptions \ref{necessaryAssumptionforRadialSol}. A similar argument using the averaging of rows instead of adding columns proves equation (\ref{eq:equForliftingThm}).  To prove the second statement, we show first that  the subspace $\ell^2_{rad}(G)$ is invariant under the matrix given in (\ref{eq:equForliftingThm}).  Let $\{ n_{rad}(x) \}_{x \in V(G)}$ be a basis vector of $\ell^2_{rad}(G)$, such that the vector components correspond to the vertices in the $n$-th transversal layer $ \pii ^{-1}(n)$ are one and zero elsewhere, i.e., 
\begin{equation}
   n_{rad}(x) =
  \begin{cases}
     1     & \quad \text{ if  $ x \in  \pii ^{-1}(n)$}, \\
    0    &    \quad \text{ otherwise. }
  \end{cases}
\end{equation}
A similar computation as in equation  (\ref{measureConditionForProof}) shows that applying the matrix given in (\ref{eq:equForliftingThm}) on $n_{rad}$ results in the following vector
\begin{equation}
   (\hamilton n_{rad})(x) =
  \begin{cases}
     a(n-1,n)     & \quad \text{ if  $ x \in  \pii ^{-1}(n-1)$}, \\
     b(n)    & \quad \text{ if  $ x \in  \pii ^{-1}(n)$}, \\
     a(n,n+1)    & \quad \text{ if  $ x \in  \pii ^{-1}(n+1)$}, \\
    0    &    \quad \text{ otherwise. }
  \end{cases}
\end{equation}
which is again in $\ell^2_{rad}(G)$. To show that the subspace $Ker  \pp $ is invariant under $\hamilton$, we proceed similarly and consider the vector $\{ n_{\bot}(x) \}_{x \in V(G)}$, given by $n_{\bot}(x)=0$ for all $x \notin   \pii ^{-1}(n)$ and satisfies the following condition
\begin{equation}
\label{orthogonalityCondition}
 \sum_{  x \in  \pii ^{-1}(n) } n_{\bot}(x)  \mu_V(x) =  0.
\end{equation}
The last equation implies that $n_{\bot}$ is orthogonal on $\ell^2_{rad}(G)$, and hence $n_{\bot} \in Ker  \pp $ due to Lemma \ref{usefulProp1}. Due to assumption \ref{necessaryAssumptionforRadialSol}, equation (\ref{orthogonalityCondition}) can be written as follows

\begin{equation}
\label{orthogonalityReformulated}
  \sum_{  y \in  \pii ^{-1}(n) }  \sum_{ (y,x ) \in E(G)} \mu_E(y,x)n_{\bot}(y) =0 , \quad \quad
\sum_{  y \in  \pii ^{-1}(n) } \sum_{(x,y ) \in E(G)} \mu_E(x,y) n_{\bot}(y) =0 .
\end{equation}
Applying the matrix given in (\ref{eq:equForliftingThm}) on $n_{\bot}$ results in the following vector
\begin{equation}
   (\hamilton n_{\bot})(x) =
  \begin{cases}
    \sum_{ (x,y ) \in E(G)}\frac{\mu_E(x,y)}{\mu_V(x)} n_{\bot}(y)  a(n-1,n)      & \quad \text{ if  $ x \in  \pii ^{-1}(n-1)$}, \\
     n_{\bot}(x)  b(n)   & \quad \text{ if  $ x \in  \pii ^{-1}(n)$}, \\
     \sum_{ (y,x ) \in E(G)} \frac{\mu_E(y,x)}{\mu_V(x)} n_{\bot}(y) a(n,n+1)   & \quad \text{ if  $ x \in  \pii ^{-1}(n+1)$}, \\
    0    &    \quad \text{ otherwise. }
  \end{cases}
\end{equation}
Equations (\ref{orthogonalityReformulated}) imply $\hamilton n_{\bot} \in Ker  \pp  $.
\end{proof}
The definition of a $\integers$-graded graph implies that two vertices $x, y \in V(G)$ in the same transversal layer can not be adjacent, i.e., $(x,y) \notin E(G)$. This would contradict $ \pii (y)= \pii (x)+ 1$. In what follows we will need the following mappings.
\begin{definition}
Given $G = (V(G),E(G),  \pii )$ a $\integers$-graded graph we define the following mappings.
\begin{enumerate}
    \item The \textit{left-hand side degree mapping} $ \mathbf{deg}_{-}: V(G) \to \nn$ is defined as follows. Assume that $x \in  \pii ^{-1}(n)$ for some $n \in \integers$, then $\mathbf{deg}_{-}(x)$ assigns the vertex $x$ the number of edges that connect $x$ to vertices in $ \pii ^{-1}(n-1)$.
    \item The \textit{right-hand side degree mapping} $\mathbf{deg}_{+}: V(G) \to \nn$  is defined as follows. Assume $x \in  \pii ^{-1}(n)$ for some $n \in \integers$, then $\mathbf{deg}_{+}(x)$  assigns the vertex $x$ the number of edges that connect $x$ to vertices in  $ \pii ^{-1}(n+1)$.
\end{enumerate}
\end{definition}

In \cite{2019arXiv190908668D,mograby2020spectra} we considered graphs with uniform transversal degrees, which is a particular case of the following assumption.   
\begin{assumption}[{Combinatorics Balance Assumption}]
\label{uniformEdgeMeasureAssumption}
Let $G = (V(G),E(G),  \pii )$ be a $\integers$-graded graph. For $x_1,y_1,x_2,y_2 \in V(G)$ such that both $x_1,y_1$ and $x_2,y_2 $ are adjacent, we set
\begin{align*}
\mu_E(x_1,y_1)=\mu_E(x_2,y_2), \  \text{if   } \  \pii (x_1)= \pii (x_2)  \  \text{ and } \   \pii (y_1)= \pii (y_2).
\end{align*}
\end{assumption}
\begin{corollary}
\label{coroForComputLift}
We impose the additional Assumption \ref{uniformEdgeMeasureAssumption} in Theorem \ref{thm:LiftingTheorem}. Then the lifted Jacobi matrix defined in (\ref{eq:equForliftingThm}) takes the form
\begin{equation*}
\label{eq:liftedAnduniformTransversal}
 \hamilton(x,y) =
  \begin{cases}
     b(n)     & \quad \text{ if } x=y \text{ and }  \pii (x)=n, \\
    \frac{1}{\mathbf{deg}_{+}(x)}  a(n,n+1) & \quad \text{ if }  (x, y) \in E(G) \text{ and }  \pii (x)=n, \\
    \frac{1}{\mathbf{deg}_{-}(x)}  a(n,n+1) & \quad \text{ if }  (y, x) \in E(G) \text{ and }  \pii (y)=n ,\\
    0    &    \quad \text{ otherwise. }
  \end{cases}
\end{equation*}
\end{corollary}
\begin{proof}
The statement follows from a direct computation. Note that Assumption \ref{uniformEdgeMeasureAssumption} implies $\mu_V(x) =  \mathbf{deg}_{+}(x) \mu_E(x,y)$, if $ (x, y) \in E(G) $ and $\mu_V(x)  =  \mathbf{deg}_{-}(x) \mu_E(y,x)$ for $ (y, x) \in E(G)$. 
\end{proof}

\section{Classical Hamiltonian Systems on $\integers$-Graded Graphs}\label{sec-CHS}
In this section, we introduce a classical Hamiltonian system on $\integers$-graded graphs. We start with some preliminary notation and assumptions.
\begin{enumerate}
	\item A vertex $x \in V(G)$ corresponds to a particle, and $\mu_V(x)$ represents its mass.
	\item Let $x \in V(G)$, then $q(x,t)$ defines the displacement of the $x$-th particle  at the moment $t$ (displacement from its equilibrium position). We will utilize the notation $q_t$ as an abbreviation for the sequence $\{ q(x,t) \}_{x \in V(G)}$ and assume $q_t \in \ell^{\infty}(V(G)) $ for all $t \in \rr$.
	\item Similarly, let $p(x,t)$ be the momentum  of the $x$-th particle. The sequence $p_t = \{ p(x,t) \}_{x \in V(G)}$ is assumed to be in $ \ell^{\infty}(V(G)) $ for all $t \in \rr$. Note that $p(x,t) = \mu_V(x) \dot{q}(x,t)$, where a superimposed dot indicates a derivative taken with respect to time $t$.
	\item We assume that $ \rr \ni t \mapsto (q_t,p_t) \in  \ell^{\infty}(V(G)) \oplus  \ell^{\infty}(V(G))$ is continuously differentiable.
	\item We introduce the difference expression $\partial : \ell^2(V(G)) \rightarrow  \ell^2(E(G)), \  f \mapsto \partial f$, where $\partial f(x,y) := f(y) - f(x), \ (x,y) \in E(G)$.
\end{enumerate}

\begin{definition}[$U$-potential]
Let $U:\rr \to \rr$  be twice continuously differentiable, i.e., $U \in C^2(\rr)$. The \textit{$U$-potential}  on a $\integers$-graded graph $G$ is defined as the  sequence $U( \partial q_t) := \{  U( \partial q_t(x,y)) \}_{(x,y) \in E(G)}$, i.e.

\begin{equation*}
E(G) \ni (x,y)  \mapsto U( \partial q_t(x,y)) = U(q(y,t)-q(x,t)) 
\end{equation*}
\end{definition}

\begin{remark}
Note that the sequence $U( \partial q_t) = \{  U( \partial q_t(x,y)) \}_{(x,y) \in E(G)}$ is bounded, i.e., $U( \partial q_t) \in \ell^{\infty}(E(G)) $. It follows by the continuity assumption on $U$ and the boundedness of $q_t$.
\end{remark}
\begin{definition}[Classical Hamiltonian]
\label{def:Classical Hamiltonian}
For a given $U$-potential, we define a Hamiltonian function on a $\integers$-graded graph $G$ formally as
\begin{align}
H_G(q_t,p_t) = \sum_{x \in V(G)}  \frac{p^2(x,t) }{2\mu_V(x)}   +   \sum_{(x,y ) \in E(G)} \mu_E(x,y) U( \partial q_t(x,y)), 
\end{align}
\end{definition}
where $(q_t,p_t) \in \ell^{\infty}(V(G)) \oplus  \ell^{\infty}(V(G))$.
We compute the corresponding equations of motion and obtain

\begin{eqnarray}
\label{eq:equationsOfMotion1}
\quad  \dot{p}(x,t) =- \frac{\partial H_G(q_t,p_t)}{\partial q(x,t)} = \sum_{(x,y ) \in E(G)} \mu_E(x,y) U'(\partial  q_t(x,y))   - \sum_{(y,x ) \in E(G)} \mu_E(y,x) U'(\partial q_t(y,x))   \\
\label{eq:equationsOfMotion2}
\dot{q}(x,t) = \frac{\partial H_G(q_t,p_t)}{\partial p(x,t)} = \frac{p(x,t) }{\mu_V(x)}.
\end{eqnarray}
{To rewrite the equations} of motion (\ref{eq:equationsOfMotion1}) and (\ref{eq:equationsOfMotion2}) as a flow on the Banach space  $\ell^{\infty}(V(G)) \oplus  \ell^{\infty}(V(G))$, we introduce the vector field

\begin{equation}
\label{vectorFieldMain}
 X(q_t, \dot{q}_t) :=  \begin{pmatrix}
   \    \left\{   \dot{q}(x,t)    \right\}_{x \in V(G)}   \\  
       \    \left\{   \sum_{(x,y ) \in E(G)} \frac{\mu_E(x,y) }{\mu_V(x)}U'(\partial  q_t(x,y))   - \sum_{(y,x ) \in E(G)} \frac{\mu_E(y,x)}{\mu_V(x)} U'(\partial q_t(y,x))  \right\}_{x \in V(G)}   
         \end{pmatrix}
\end{equation}
Then, the equations of motion (\ref{eq:equationsOfMotion1}) and (\ref{eq:equationsOfMotion2})  take the form,

\begin{equation}
	\label{dynamicalSystem}
	\frac{d}{dt}
        \begin{pmatrix}
            q_t \\ \dot{q}_t
         \end{pmatrix}
          = X(q_t, \dot{q}_t).
\end{equation}

\begin{remark}
Note that we will assume the existence of a  $ \delta >0$ such that $\mu_V(x) \geq \delta $ for all $x \in V(G)$. Equation  (\ref{eq:equationsOfMotion2}) implies $\dot{q}(x,t)  = \frac{p(x,t) }{\mu_V(x)} \leq \frac{p(x,t) }{\delta}$ and hence  $\dot{q}_t \in \ell^{\infty}(V(G)) $ for all $t \in \rr$. By the continuity assumption on $U'$ we can verify $ X(q_t, \dot{q}_t)  \in  \ell^{\infty}(V(G)) \oplus  \ell^{\infty}(V(G))$.
\end{remark}

\begin{theorem}
\label{existanceUniqueTheorem}
	If we assume that there exists $ \delta >0$ such that $\mu_V(x) \geq \delta $ for all $x \in V(G)$ then, given an initial value $(q_0, \dot{q}_0) \in  \ell^{\infty}(V(G)) \oplus  \ell^{\infty}(V(G)) $, there exists a unique integral curve of $X$ at $(q_0, \dot{q}_0) $, i.e., 	
\begin{equation*}
c: I \to \ell^{\infty}(V(G)) \oplus  \ell^{\infty}(V(G)), \ \  t \mapsto (q_t, \dot{q}_t),
\end{equation*}	
where $I$ is an interval, $0 \in I$, such that $\frac{d}{dt} c(t) = X(c(t))$   for all $t \in I$, and  $\ c(0)  = (q_0, \dot{q}_0)$.
\end{theorem}
\begin{proof}
The reader is referred to Lemma 4.1.6 (page 242), Lemma 4.1.9 (page 244) and Theorem 4.1.11 (page 246) in \cite{amr}. It suffices to show that the vector field $X$ is locally Lipschitz. Observe that $U'$ is locally Lipschitz due to the assumption $U \in C^2(\rr)$. Hence, choose $q_t$ and $\tilde{q}_t$ such that

\begin{equation}
|U'(\partial q_t(x,y)) - U'(\partial \tilde{q}_t(x,y))  | \leq K_1 |\partial q_t(x,y) - \partial \tilde{q}_t(x,y)|
\end{equation}
holds for all $(x,y) \in E(G)$. This implies, $|U'(\partial q_t(x,y)) - U'(\partial \tilde{q}_t(x,y))  | \leq 2 K_1 ||q_t-\tilde{q}_t||_{\ell^{\infty}(V(G))}$ for all edges $(x,y) \in E(G)$. Moreover, for every $x \in V(G)$, we have $ \frac{1}{\mu_V(x)} \leq \frac{1}{ \delta}$, $ \  \sum_{(x,y ) \in E(G)} \mu_E(x,y) \leq 1$ and $ \sum_{(y,x ) \in E(G)} \mu_E(y,x) \leq 1$ (sum over $y$). Hence, $\frac{4 K_1}{\delta}||q_t-\tilde{q}_t||_{\ell^{\infty}(V(G))}$ provides an upper bound for the second component of $X$. Let $K= max \{\frac{4 K_1}{\delta}, 1 \}$, then
\begin{equation}
||   X(q_t, \dot{q}_t) -  X(\tilde{q}_t, \dot{\tilde{q}}_t)  || \leq K \left( \ ||q_t-\tilde{q}_t||_{\ell^{\infty}(V(G))} + ||\dot{q}_t-\dot{\tilde{q}}_t||_{\ell^{\infty}(V(G))} \ \right)
\end{equation}
\end{proof}
\begin{example}
\label{exampleTrivialZgraph}
The simplest $\integers$-graded graph is $\integers$ itself considered as a path graph  $(V(\integers), E(\integers))$ with the vertex set $V(\integers) = \integers$  and an edge set $E(\integers) =\{(n-1,n) : \ n \in \integers \}$. The identity function plays the role of the rank function $ \pii $. Each transversal layer contains a single vertex, $  \pii ^{-1}(n) = \{n \}$ and as a probability measure on a transversal layer, we have $\mu_V(n)=1$ for each $n \in \integers$. Similarly, we have $\mu_E(n,n+1)=1$ for all edges $(n,n+1) \in E(\integers)$. Every function in $\ell^2(\integers)$ is radial, i.e., $\ell^2_{rad}(\integers) = \ell^2(\integers) $ and definition \ref{def:Classical Hamiltonian} provides the following Hamiltonian function

\begin{equation}
H_{\integers}(q_t,p_t) = \sum_{n \in \integers}  \frac{p_{\integers}^2(n,t) }{2}   +   \sum_{(n,n+1 ) \in E(\integers)}  U( q_{\integers}(n+1,t)-q_{\integers}(n,t)).
\end{equation}
The equations of motion that govern this one-dimensional system are given by
\begin{equation}
\label{eq: equationsOfMotiononZ}
\dot{p}_{\integers}(n,t)= U'(q_{\integers}(n+1,t)-q_{\integers}(n,t))-U'(q_{\integers}(n-1,t)-q_{\integers}(n,t)),  \quad \dot{q}_{\integers}(n,t) = p_{\integers}(n,t).
\end{equation}
For a given $(q_{\integers,0}, \dot{q}_{\integers,0}) \in  \ell^{\infty}(V(\integers)) \oplus  \ell^{\infty}(V(\integers))$, we apply Theorem \ref{existanceUniqueTheorem} and solve the equations of motion (\ref{eq: equationsOfMotiononZ}). For later use, we denote the solution by

\begin{equation}
\label{eq:ZsolutionToLift}
c_{\integers}: I \to \ell^{\infty}(V(\integers)) \oplus  \ell^{\infty}(V(\integers)), \ \  t \mapsto c_{\integers}(t):=\{ q_{\integers}(n,t), \dot{q}_{\integers}(n,t) \}_{n \in \integers},  \ \  c_{\integers}(0)=(q_{\integers,0}, \dot{q}_{\integers,0}).
\end{equation}
\end{example}

We lift the curve $c_{\integers}$ given in (\ref{eq:ZsolutionToLift})  to a $\integers$-graded graph $G$ in the following sense. 
\begin{definition}
\label{def:liftToRadialSol}
Let  $ \pp ^\ast$ be the adoint of the averaging operator $ \pp $. We lift the one-dimensional solution given in (\ref{eq:ZsolutionToLift}) to a radial curve $c_{rad}: I \to \ell^{\infty}(V(G)) \oplus  \ell^{\infty}(V(G))$ defined by
\begin{equation}
\label{eq:GsolutionLifted}
 t \mapsto c_{rad}(t):=\{  \pp ^\ast (q_{\integers}(.,t))(x),  \pp ^\ast (\dot{q}_{\integers}(.,t)(x) \}_{x \in V(G)}.
\end{equation}
Note that  $c_{rad}(0)=\{  \pp ^\ast (q_{\integers}(.,0))(x),  \pp ^\ast (\dot{q}_{\integers}(.,0)(x) \}_{x \in V(G)}=\left(  \pp ^\ast (q_{\integers,0}),  \pp ^\ast (\dot{q}_{\integers,0}) \right)$. It is easily seen that $c_{rad}(t) \in \ell_{rad}^{\infty}(V(G)) \oplus  \ell_{rad}^{\infty}(V(G))$ for all $t \in I$, as the range of $ \pp ^{\ast}$ is $\ell^2_{rad}(G)$ (see Lemma \ref{usefulProp1}).
\end{definition}
\begin{theorem}
\label{thm:liftToradailSolution}
Under the Assumption \ref{necessaryAssumptionforRadialSol} the curve $c_{rad}(t)$, defined in (\ref{eq:GsolutionLifted}) on an interval $I\ni0$, uniquely solves the initial value problem given in equation (\ref{dynamicalSystem}), i.e $\frac{d}{dt} c_{rad}(t) = X(c_{rad}(t))$ for all $ t \in I$, and  $\ c_{rad}(0)  = \left(  \pp ^\ast (q_{\integers,0}),  \pp ^\ast (\dot{q}_{\integers,0}) \right)$. In particular,  if $c(t)$ is a solution of the initial value problem and $c(0) \in   \ell_{rad}^{\infty}(V(G)) \oplus  \ell_{rad}^{\infty}(V(G))  $, then  $c(t) \in   \ell_{rad}^{\infty}(V(G)) \oplus  \ell_{rad}^{\infty}(V(G))$ for all  $t \in I$.
\end{theorem}
\begin{proof}
Let $x \in  \pii ^{-1}(n)$, then $ \partial  \pp ^\ast (q_{\integers}(.,t))(x,y) :=  \pp ^\ast (q_{\integers}(.,t))(y) -  \pp ^\ast (q_{\integers}(.,t))(x) = q_{\integers}(n+1,t) - q_{\integers}(n,t)$. In particular,
\begin{equation}
U'(  \partial  \pp ^\ast (q_{\integers}(.,t))(x,y)) = U'(q_{\integers}(n+1,t) - q_{\integers}(n,t) ).
\end{equation}
Assumptions \ref{necessaryAssumptionforRadialSol} gives $1 =  \frac{1}{\mu_V(x)}\sum_{ (x,y ) \in E(G)} \mu_E(x,y)$. It follows,
\begin{equation}
\sum_{ (x,y ) \in E(G)}  \frac{\mu_E(x,y)}{\mu_V(x)} U'(  \partial  \pp ^\ast (q_{\integers}(.,t))(x,y)) = U'(q_{\integers}(n+1,t) - q_{\integers}(n,t) ).
\end{equation}
Similarly, we verify
\begin{equation}
\sum_{ (y,x ) \in E(G)}  \frac{\mu_E(y,x)}{\mu_V(x)} U'(  \partial  \pp ^\ast (q_{\integers}(.,t))(y,x)) = U'(q_{\integers}(n-1,t) - q_{\integers}(n,t) ).
\end{equation}
Hence, this shows that the second component in the vector field (\ref{vectorFieldMain}) is equivalent to the right-hand side of the first equation in (\ref{eq: equationsOfMotiononZ}). To verify the left-hand side, note that $\dot{q}_{\integers}(n,t) = p_{\integers}(n,t)$.
\end{proof}

\section{Toda Lattices on $\integers$-Graded Graphs}\label{sec-Toda}
Our primary goal is to extend the definition of a one-dimensional Toda lattice to a $\integers$-graded graph $G$. To this end, we introduce the Flaschka's variables \cite{10.1143/PTP.51.703, PhysRevB.9.1924}. These variables represent a transformation of the phase space variables $(q_t,p_t)$ to coordinates, which give the dynamics in the Lax pair formalism.
\begin{definition}
\label{FlaschkaVariables}
For a given   $U$-potential the Flaschka transformation of $q_t,p_t \in  \ell^{\infty}(V(G)) $ is defined by
\begin{align}
\label{eq:flaschka1}
 a_t(x,y) := \frac{1}{2} U'(\partial q_t(x,y)/2)=\frac{1}{2} U' \left(    \frac{q(y,t)}{2}-\frac{q(x,t)}{2} \right),                                          \quad   \quad  (x,y) \in E(G), \\
\label{eq:flaschka2}
  b(x,t) := - \frac{p(x,t)}{2\mu_V(x)},  \quad \quad  x \in V(G).
\end{align}
Due to the continuity of $U'$ as well as boundedness of $q_t$ and $p_t$, we imply  $a_t = \{  a_t(x,y) \}_{(x,y) \in E(G)} \in \ell^{\infty}(E(G))$ and  $b_t = \{ b(x,t) \}_{x \in V(G)} \in \ell^{\infty}(V(G)) $. We call $a_t$ and $b_t$ Flaschka's variables.
\end{definition}
\begin{proposition}
	We assume that the restriction of $U'$ to its range, i.e., $U'|_{range(U')}: range(U') \to \rr$ is injective. In the Flaschka's variables, the equations of motion (\ref{eq:equationsOfMotion1}) and (\ref{eq:equationsOfMotion2})  take the form
	\begin{align}
	\label{eq:propForFlow2}
	 \dot{a}_t(x,y) = -\frac{1}{2} U''\circ (U')^{-1}\left( 2 a_t(x,y)\right)  \  \partial  b_t(x,y),
	\\
	\label{eq:propForFlow1}
	\dot{b}(x,t) =  \sum_{(y,x ) \in E(G)} \frac{\mu_E(y,x)}{2\mu_V(x)}  U'[2 (U')^{-1} (2a_t(y,x))]  - \sum_{ (x,y ) \in E(G)}  \frac{\mu_E(x,y)}{2\mu_V(x)} U'[2 (U')^{-1} (2a_t(x,y))].
	\end{align}
\end{proposition}

\begin{proof}
Note that equation (\ref{eq:equationsOfMotion2}) implies  $\frac{\dot{q}(x,t)}{2} = \frac{p(x,t)}{2\mu_V(x)} = -  b(x,t)$. We observe that $a_t(x,y) \in range(U')$ and hence equation (\ref{eq:flaschka1}) is equivalent to $\frac{\partial q_t(x,y)}{2} =  (U')^{-1} \left( 2 a_t(x,y) \right) $.  Differentiating equation (\ref{eq:flaschka1}) with respect to $t$, we obtain
\begin{equation}
\dot{a}_t(x,y) = \frac{1}{2} U''(\partial q_t(x,y)/2)  \left(    \frac{\dot{q}(y,t)}{2}-\frac{\dot{q}(x,t)}{2} \right)=- \frac{1}{2} U''\circ (U')^{-1}( 2 a_t(x,y))  \left( b(y,t) - b(x,t) \right).
\end{equation}
Using $\partial b_t(x,y)= b(y,t) - b(x,t)$ implies equation (\ref{eq:propForFlow2}).  Similarly, we differentiate equation (\ref{eq:flaschka2}) with respect to $t$ and substitute equation (\ref{eq:equationsOfMotion1}). We obtain
\begin{equation}
\dot{b}(x,t) = - \frac{\dot{p}(x,t)}{2\mu_V(x)} =  \sum_{(y,x ) \in E(G)} \frac{\mu_E(y,x)}{2\mu_V(x)}  U'(\partial  q_t(y,x))  - \sum_{ (x,y ) \in E(G)} \frac{\mu_E(x,y)}{2\mu_V(x)}  U'(\partial q_t(x,y)).
\end{equation}
We observe $U'(\partial q_t(x,y)) =U'[2 (U')^{-1} (2a_t(x,y))]$ and hence imply equation (\ref{eq:propForFlow1}).
\end{proof}
In our investigations on Hamiltonians of particle systems with exponential interactions we follow  \cite[Chapter 12]{TeschlbookJacobi} 
and introduce the potential
\begin{equation}
\label{eq:TodaPotential}
U(r) := e^{-r} -1.
\end{equation}

Note that the potential defined in (\ref{eq:TodaPotential}) is a rescaled variant of the original potential introduced by M.~Toda. For details the reader is referred to  \cite[page 223]{TeschlbookJacobi}.
A Toda lattice on a $\integers$-graded graph $G$ is given by the Hamiltonian

\begin{equation}
\label{eq:TodaHamilGeneralZgraph}
H_G(q_t,p_t) = \sum_{x \in V(G)}  \frac{p^2(x,t) }{2\mu_V(x)}   +   \sum_{(x,y ) \in E(G)} \mu_E(x,y) \left(  e^{-(q(y,t)-q(x,t))} -1  \right).
\end{equation}
The Toda equations of motion represented in the phase space coordinates  are given by
\begin{eqnarray}
\label{eq:equationsOfMotiontoda1}
\quad  \dot{p}(x,t) = \sum_{(y,x ) \in E(G)} \mu_E(y,x) e^{-(q(x,t)-q(y,t))}  -  \sum_{(x,y ) \in E(G)} \mu_E(x,y) e^{-( q(y,t)-q(x,t))  }  ,  \\
\label{eq:equationsOfMotiontoda2}
\dot{q}(x,t) = \frac{p(x,t) }{\mu_V(x)}.
\end{eqnarray}
The same equation when expressed in the Flaschka coordinates are given by

\begin{eqnarray}
\label{eq: TransformedequationsOfMotion}
 \begin{cases}
 \     \dot{a}_t(x,y)  = a_t(x,y) (b(y,t)-  b(x,t)),    \\
  \     \dot{b}(x,t) =   \sum_{(x,y ) \in E(G)} \frac{\mu_E(x,y)}{\mu_V(x)} 2 a^2_t(x,y)  - \sum_{ (y ,x) \in E(G)}  \frac{\mu_E(y,x)}{\mu_V(x)} 2 a^2_t(y,x).  
  \end{cases}
\end{eqnarray} 

\begin{example}[Continuation of Example \ref{exampleTrivialZgraph}]
The Hamiltonian $H_{\integers}(q,p)$ in this case describes nothing else but the one-dimensional Toda lattice  \cite[Equation (12.3), page 223]{TeschlbookJacobi}. The equations of motion are transformed into the following form
\begin{eqnarray}
\label{eq: TransformedequationsOfMotion1D}
 \begin{cases}
 \     \dot{a}_{\integers,t}(n,n+1)  = a_{\integers,t}(n,n+1) (b_{\integers}(n+1,t)-  b_{\integers}(n,t)),    \\
  \     \dot{b}_{\integers}(n,t) =   2 a^{2}_{\integers,t}(n,n+1) - 2 a^{2}_{\integers,t}(n-1,n) .   
  \end{cases}
\end{eqnarray} 
Reformulating equations (\ref{eq: TransformedequationsOfMotion1D}) as a Lax Pair and employing arguments from the inverse spectral theory, one can write the general N-soliton solution \cite[Equation (12.16), page 225]{TeschlbookJacobi},

\begin{equation}
\label{NSolitonSol}
q_N(n,t) = q_0 - \ln{ \frac{\det(\mathbb{1}+C_N(n,t))}{\det(\mathbb{1}+C_N(n-1,t))}},
\end{equation}
where
\begin{equation}
C_N(n,t) = \left( \frac{\sqrt{\gamma_i \gamma_j}}{1-e^{-(\kappa_i + \kappa_j)}} e^{-(\kappa_i + \kappa_j)n -(\sigma_i \sinh{(\kappa_i)}+ \sigma_j \sinh{(\kappa_j)})t     } \right)_{1 \leq i,j \leq N}
\end{equation}
and $\kappa_j, \gamma_j > 0$, $\sigma_j \in \{\pm 1\}$.
\end{example}
\begin{corollary}
\label{globalSolutionsAndRadial}
Let $U$ be a Toda potential given in (\ref{eq:TodaPotential}). For such a potential, the solution given in Theorem \ref{thm:liftToradailSolution} is global, i.e., it exists for all $ \ t \in \rr$.  Moreover, we lift the N-soliton solution $\{q_N(n,t)\}_{n \in \integers}$ given in (\ref{NSolitonSol})  to a radial solution on a $\integers$-graded graph $G$ in the sense of definition \ref{def:liftToRadialSol}. We set  $q_{G,N}: \rr \to \ell^{\infty}(V(G)) \oplus  \ell^{\infty}(V(G))$, $ \  q_{G,N}(t):=\{  \pp ^\ast (q_{N}(\cdot,t))(x),  \pp ^\ast (\dot{q}_{N}(\cdot,t)(x) \}_{x \in V(G)}$. Then $q_{G,N}$ defines an N-soliton solution on the  $\integers$-graded graph $G$.
\end{corollary}
\begin{proof}
The first statement follows by \cite[Theorem 12.6, page 232]{TeschlbookJacobi}. The second statement is verified with a similar argument as in the proof of Theorem \ref{thm:liftToradailSolution}.
\end{proof}

\begin{example}[Figure~\ref{fig:simplestZgradedGraph}]
\label{ex:MainExample}
We consider the simplest nontrivial infinite $\integers$-graded graph $G$ displayed in Figure~\ref{fig:simplestZgradedGraph} on page~\pageref{fig:simplestZgradedGraph}. As indicated in the figure, we denote the vertices by $V(G) = \{ \dots,-2,-1,0w_1,0w_2,1,2,\dots \}$. The transversal layers are $ \pii ^{-1}(0) = \{0w_1,0w_2\} \text{  and  } \  \pii ^{-1}(n) = \{n\}  \  \ \forall n \in \integers \backslash \{0\}$. For the measure, we have $\mu_V(x)=1$ for $x \in  V(G) \backslash \{ 0w_1,0w_2 \} $  and we set $\mu_V(x)=\frac{1}{2}$ for $x \in  \{ 0w_1,0w_2 \} $. The corresponding averaging operator $ \pp $, its adjoint  $ \pp ^{\ast}$, and the projector $\proj  $ are computed and given in the appendix. Similarly, for the measure on the edges, we have $\mu_E(n,n+1)=1$ for all $n \in \integers \backslash \{0,-1\}$ and we set $\mu_E(-1,0w_1) = \mu_E(-1,0w_2) = \mu_E(0w_1,1) =\mu_E(0w_2,1)=\frac{1}{2} $. In particular, the assumptions \ref{necessaryAssumptionforRadialSol} are satisfied, i.e.
\begin{align}
 \mu_V(-1)=\mu_E(-1,0w_1) + \mu_E(-1,0w_2), \quad \mu_V(0w_i)=\mu_E(0w_i,1) ,  \nonumber \\
 \mu_V(1)=\mu_E(0w_1,1) + \mu_E(0w_2,1), \quad \mu_V(0w_i)=\mu_E(-1,0w_i) ,  \nonumber
\end{align}
 for $i \in \{1,2\}$. A direct computation of the equations of motion (\ref{eq: TransformedequationsOfMotion}) expressed in the Flaschka coordinates gives for $n \in \integers \backslash \{-1,0,1\}$
\begin{equation}
\label{eq: TransformedequationsOfMotiona1}
 \begin{cases}
 \     \dot{a}_t(n,n+1)  = a_t(n,n+1) (b(n+1,t)-  b(n,t)),    \\
  \     \dot{b}(n,t) =   2 a^2_t(n,n+1)  -  2 a^2_t(n-1,n).
  \end{cases}
\end{equation} 
For the vertices $-1$, $1$, we obtain
\begin{equation}
\label{eq: TransformedequationsOfMotiona2}
 \begin{cases}
 \     \dot{a}_t(-1,0w_i)  = a_t(-1,0w_i) (b(0w_i,t)-  b(-1,t)),    \\
  \     \dot{b}(-1,t) =     a^2_t(-1,0w_1) + a^2_t(-1,0w_2)  -  2 a^2_t(-2,-1),
  \end{cases}
  \
   \begin{cases}
 \     \dot{a}_t(1,2)  = a_t(1,2) (b(2,t)-  b(1,t)),    \\
  \     \dot{b}(1,t) =    2 a^2_t(1,2) - a^2_t(0w_1,1) - a^2_t(0w_2,1),
  \end{cases}
\end{equation} 
with $i \in \{1,2\}$. Similarly, for the vertices $0w_i$, we have
\begin{equation}
\label{eq: TransformedequationsOfMotiona3}
 \begin{cases}
 \     \dot{a}_t(0w_i,1)  = a_t(0w_i,1) (b(1,t)-  b(0w_i,t)),    \\
  \     \dot{b}(0w_i,t) =   2 a^2_t(0w_i,1)  -  2 a^2_t(-1,0w_i).  
  \end{cases}
\end{equation} 
with $i \in \{1,2\}$. 

This example continues in Example~\ref{ex:ContinuationOfMainExample}.
\end{example}

\section{Non-existence of   Lifted Lax Pairs}
\label{sec:LiftingLaxPair}
The Lax approach starts with the observation that we can find two  time-dependent operators acting on $\ell^2(\integers)$, i.e., a \textit{Lax pair} $\{\PlaxZ  (t), \LhamiltonZ(t)\}$, such that the following system of  differential equations
\begin{equation}
\label{eq:firstLaxEq}
\frac{d}{dt} \LhamiltonZ(t) = [\PlaxZ  (t), \LhamiltonZ(t)]:=\PlaxZ  (t) \LhamiltonZ(t)- \LhamiltonZ(t)\PlaxZ  (t), \quad t \in \rr. 
\end{equation}
is equivalent to (\ref{eq: TransformedequationsOfMotion1D}). The matrix $ \LhamiltonZ(t)$ is the same infinite Jacobi matrix as in \eqref{eq:infiniteJacobiMatrix-},
\begin{equation}
\label{eq:infiniteJacobiMatrix}
\LhamiltonZ(t):    \ell^{2}(\integers) \to \ell^{2}(\integers), \quad   \quad 
	  \varphi(n)  \mapsto a_{\integers,t}(n,n+1)  \varphi(n+1) + b_{\integers}(n,t )  \varphi(n) + a_{\integers,t}(n-1,n)  \varphi(n-1),
\end{equation}
and  $\PlaxZ  (t) := [\LhamiltonZ(t)]_{+} - [\LhamiltonZ(t)]_{-}$, where  $[\LhamiltonZ(t)]_{+}$  and $ [\LhamiltonZ(t)]_{-}$ define the upper and lower triangular parts of $\LhamiltonZ(t)$ respectively.  Equation (\ref{eq:firstLaxEq}) is called the \textit{Lax equation} corresponding to $\{\PlaxZ  (t), \LhamiltonZ(t)\}$ and we refer to $\PlaxZ  (t)$ and $\{\PlaxZ  (t), \LhamiltonZ(t)\}$ as a Lax operator and a Lax Pair, respectively. A crucial aspect of the Lax method is that the dynamics of the system (\ref{eq:firstLaxEq}) evolve $\LhamiltonZ(t)$ in such a way that its spectrum is invariant.  For this choice $\{\PlaxZ  (t), \LhamiltonZ(t)\}$, it can be shown that the Lax equation (\ref{eq:firstLaxEq}) indeed reproduces (\ref{eq: TransformedequationsOfMotion1D}). Note that $\PlaxZ  (t) $ is skew-adjoint for $ t \in \rr$. We lift $\LhamiltonZ(t)$ to an operator on the $\integers$-graded graph $G$ in the sense of Definition \ref{def:LiftOperator}. We impose Assumption  \ref{necessaryAssumptionforRadialSol} on the measures $\mu_V$ and $\mu_E$ to justify this lift (see Theorem \ref{thm:LiftingTheorem}). Let $\{\Lhamilton(t)\}_{t \in \rr}$ be a family of lifted Jacobi operators on $G$. 
\begin{definition}
\label{def:radLax}
The radial Lax operator on $G$ is defined as $\Plax_{rad}(t): \ell^2(G)  \to \ell^2(G), \ \Plax_{rad}(t) :=  \pp ^{\ast}    \PlaxZ  (t)  \pp $.
\end{definition}

\begin{remark}
Note that $\Plax_{rad}(t)$ in general \emph{does not  reflect  the adjacency relation of the graph} $G$ in the sense of 
Definition~\ref{def:LiftOperator}.  
\end{remark}

\begin{proposition}
\label{prop:RadialLax}
Let $\{\Lhamilton(t)\}_{t \in \rr}$ be a family of lifted Jacobi operators. Suppose that Assumption \ref{separationAnsatz}, the Spectral Separation Assumption, holds for each $\Lhamilton(t)$ and that  $\LhamiltonZ(t)$ satisfies the Lax equation (\ref{eq:firstLaxEq}). Then, the following equation holds for all $ t \in \rr$.
\begin{equation}
\label{radialLaxEquation}
 \proj \frac{d}{dt} \Lhamilton (t) \proj   = \frac{d}{dt} \Lhamilton (t) \proj   = [\Plax_{rad}(t) ,\  \Lhamilton(t) ] 
\end{equation}
and 
\begin{equation}
	\label{eq:LaxForH}
	\frac{d}{dt} \Lhamilton (t)  = \frac{d}{dt} \Lhamilton (t) \proj   + \frac{d}{dt} \Lhamilton (t)(\Id - \proj  ) = [\Plax_{rad}(t) ,\  \Lhamilton(t) ]+ \frac{d}{dt} \Lhamilton (t)(\Id - \proj  ).
\end{equation}
\end{proposition}
\begin{proof}
We observe
\begin{eqnarray}
\label{eq:calForTheoremRadialLax}
  \pp  \frac{d}{dt}\Lhamilton(t)  \pp ^{\ast} = \frac{d}{dt}\LhamiltonZ(t)  &=& \PlaxZ  (t) \LhamiltonZ(t) - \LhamiltonZ(t) \PlaxZ  (t), \nonumber \\
					  &=&  \pp  \left(  \pp ^{\ast} \PlaxZ  (t)  \pp  \right) \proj   \left(  \pp ^{\ast} \LhamiltonZ(t)  \pp  \right)  \pp ^{\ast} -  \pp  \left( \pp ^{\ast} \LhamiltonZ(t)  \pp  \right) \left(  \pp ^{\ast}\PlaxZ  (t)  \pp  \right) \proj    \pp ^{\ast},  \nonumber  \\
					  &=&  \pp  \left( \Plax_{rad}(t) \proj   \right)  \left( \Lhamilton(t)  \proj   \right)  \pp ^{\ast} -  \pp  \left( \Lhamilton(t)  \proj   \right) \left(  \Plax_{rad}(t) \proj   \right)  \pp ^{\ast},  \nonumber \\
					  &=&  \pp  [ \Plax_{rad}(t) \proj   , \ \Lhamilton(t)  \proj   ]  \pp ^{\ast}.
\end{eqnarray} 
The second equality holds as $\LhamiltonZ(t)$ satisfies the Lax equation (\ref{eq:firstLaxEq}). The third equality holds by Lemma \ref{usefulProp1}, ($range ( \pp ^{\ast})=\ell^2_{rad}(G)$). The fourth equality holds as $\Lhamilton(t)$  satisfies the assumption \ref{separationAnsatz} (Spectral Separation Assumption) and the following observation $ \pp ^{\ast} \LhamiltonZ(t)  \pp  =  \pp ^{\ast}  \pp  \Lhamilton(t)  \pp ^{\ast}  \pp  = \proj   \Lhamilton(t)  \proj   =  \Lhamilton(t)  \proj   $. We right-multiply the equation (\ref{eq:calForTheoremRadialLax}) with $ \pp $. Using $ \pp ^{\ast}  \pp  = \proj  $, we obtain
\begin{equation}
  \pp  \frac{d}{dt}\Lhamilton(t) \proj   =  \pp  [ \Plax_{rad}(t) \proj   , \ \Lhamilton(t)  \proj   ].
\end{equation}
Hence, for an arbitrary $\varphi \in \ell^2(G)$, we have $\frac{d}{dt}\Lhamilton(t) \proj   \ \varphi -  [ \Plax_{rad}(t) \proj   , \ \Lhamilton(t)  \proj   ] \ \varphi  \in Ker \pp $. On the other hand, we can verify that $\frac{d}{dt}\Lhamilton(t) \proj   \ \varphi -  [ \Plax_{rad}(t) \proj   , \ \Lhamilton(t)  \proj   ] \ \varphi  \in \ell^2_{rad}(G)$,
therefore $\frac{d}{dt}\Lhamilton(t) \proj   =  [ \Plax_{rad}(t) \proj   , \ \Lhamilton(t)  \proj   ]$.   This means that we have effectively constructed a Lax pair  restricted to the projected, or radial, subspace, but not on the entire space of functions. 

Recall $\ell^2(G)= \ell^2_{rad}(G) \oplus Ker  \pp $ and define $\proj_{\bot}: \ell^2(G) \to Ker \pp $. As $\Plax_{rad}(t) \proj_{\bot}=0$, it follows $ [ \Plax_{rad}(t) \proj_{\bot} , \ \Lhamilton(t) ]=0$. Moreover, the Assumption \ref{separationAnsatz} implies $[ \Plax_{rad}(t) \proj   , \ \Lhamilton(t)  \proj_{\bot}]=0$.
\end{proof}
As a consequence of the skew-adjointness of $\PlaxZ  (t)$, we have the following lemma. 
\begin{lemma}
\label{skewAdjointness}
 $ \Plax_{rad}(t)$ is a skew-adjoint  bounded operator for all $t \in \rr$.
\end{lemma}
\begin{proof}
The boundedness follows by the boundedness of $ \PlaxZ  (t)$. Moreover, we have $$ \bra{ \Plax_{rad}(t) \psi}\ket{ \varphi}_{G} =   \bra{ \PlaxZ  (t)  \pp  \psi}\ket{ \pp   \varphi}  = - \bra{  \pp  \psi}\ket{ \PlaxZ  (t)   \pp   \varphi}= -  \bra{  \psi}\ket{ \Plax_{rad}(t) \varphi}_{G} ,    \  \forall \ \psi, \varphi \in \ell^2(G).$$
\end{proof}

In particular, we raise the question of whether we can find a skew-adjoint operator $\Plax_{\bot}(t)$ such that $\frac{d}{dt} \Lhamilton (t)(I- \proj  ) = [\Plax_{\bot}(t) ,\  \Lhamilton(t) ]$. In such a  case, equation (\ref{eq:LaxForH}) would give $\frac{d}{dt} \Lhamilton (t) = [\PlaxG (t) ,\  \Lhamilton(t) ]$, where we set $\PlaxG (t)=\Plax_{rad}(t) + \Plax_{\bot}(t)$ and consequently imply that the pair $\{\PlaxG (t), \Lhamilton(t)\}$ satisfies the isospectral property. The following example shows that this is not always possible with $\Lhamilton(t)$ and  $ \Plax_{\bot}(t)$  reflecting the adjacency relation of the graph  $G$.

\begin{example}[Continuation of Example \ref{ex:MainExample}, Figure~\ref{fig:simplestZgradedGraph}]
\label{ex:ContinuationOfMainExample}
Under   assumption  \ref{uniformEdgeMeasureAssumption}, using Corollary \ref{coroForComputLift}, we can compute the lifted Jacobi operator $\Lhamilton(t)$.
The operator $\Lhamilton(t)$ is almost identical with $\LhamiltonZ(t)$, except at the rows and columns corresponding to the vertices $\{-1,0w_1,0w_2,1\}$. The corresponding submatrix is indicated inside a rectangle in the middle of the matrix.

\begin{equation}
\label{exampleHamilton}
  \renewcommand{\arraystretch}{1.5}
  \Lhamilton(t)=
  \left(
  \begin{array}{c c c  c c c c  c c c }
     &  \ddots   & \vdots    &  \vdots & \vdots &  \vdots  &  \vdots  & \vdots & \iddots  \\
 \dots   & b_{\integers}(-2,t) & a_{\integers, t}(-2,-1) & 0 & 0 & 0 & 0  & 0  & \dots \\
   \cline{3-6} 
\dots    &  \multicolumn{1}{c|}{a_{\integers, t}(-2,-1)} & b_{\integers}(-1,t) & \frac{a_{\integers, t}(-1,0)}{2} & \frac{a_{\integers, t}(-1,0)}{2} & \multicolumn{1}{c|}{0} & 0 & 0 & \dots \\
 \dots    & \multicolumn{1}{c|}{0} & a_{\integers, t}(-1,0) &  b_{\integers}(0,t) & 0 &  \multicolumn{1}{c|}{a_{\integers, t}(0,1)}  & 0 & 0 & \dots \\
 \dots     & \multicolumn{1}{c|}{0} &  a_{\integers, t}(-1,0) & 0 &  b_{\integers}(0,t) &  \multicolumn{1}{c|}{a_{\integers, t}(0,1)}  & 0 & 0 & \dots \\
  \dots  & \multicolumn{1}{c|}{0}   & 0 & \frac{a_{\integers, t}(0,1)}{2} & \frac{a_{\integers, t}(0,1)}{2} & \multicolumn{1}{c|}{ b_{\integers}(1,t)} &  a_{\integers, t}(1,2) & 0 & \dots \\ 
   \cline{3-6} 
  \dots  & 0 & 0 & 0 & 0 & a_{\integers, t}(1,2) &  b_{\integers}(2,t) & a_{\integers, t}(2,3) &  \\
   \iddots  & \vdots & \vdots & \vdots &  \vdots &  & a_{\integers, t}(2,3) & \ddots\\
  \end{array} 
  \right) 
\end{equation}

 We compute $\frac{d}{dt} \Lhamilton (t) \proj   = [\Plax_{rad}(t) ,\  \Lhamilton(t) ]$:
\begin{equation*}
  \renewcommand{\arraystretch}{1.7}
  \frac{d}{dt}\Lhamilton \proj   =
  \left(
  \begin{array}{c c c  c c c c  c c c }
     &  \ddots     &  \vdots & \vdots &  \vdots  &  \vdots   & \iddots  \\
   \dots &   \dot{b}_{\integers}(-2,t) &\dot{a}_{\integers,t}(-2,-1) & 0 & 0 & 0 & 0  &  \dots \\
   \cline{3-6} 
\dots   &  \multicolumn{1}{c|}{\dot{a}_{\integers,t}(-2,-1)} & \dot{b}_{\integers}(-1,t) & \frac{\dot{a}_{\integers,t}(-1,0)}{2}&\frac{\dot{a}_{\integers,t}(-1,0)}{2} & \multicolumn{1}{c|}{0} & 0 &  \dots \\
 \dots     & \multicolumn{1}{c|}{0} & \dot{a}_{\integers,t}(-1,0) &  \frac{ \dot{b}_{\integers}(0,t)}{2} &  \frac{ \dot{b}_{\integers}(0,t)}{2} &  \multicolumn{1}{c|}{\dot{a}_{\integers,t}(0,1)}  & 0 &  \dots \\
 \dots     & \multicolumn{1}{c|}{0} &  \dot{a}_{\integers,t}(-1,0) &  \frac{ \dot{b}_{\integers}(0,t)}{2} &   \frac{ \dot{b}_{\integers}(0,t)}{2} &  \multicolumn{1}{c|}{\dot{a}_{\integers,t}(0,1)}  & 0 &  \dots \\
  \dots  & \multicolumn{1}{c|}{0}   & 0 & \frac{\dot{a}_{\integers,t}(0,1)}{2} & \frac{\dot{a}_{\integers,t}(0,1)}{2}& \multicolumn{1}{c|}{ \dot{b}_{\integers}(1,t)} &  \dot{a}_{\integers,t}(1,2) &  \dots \\ 
   \cline{3-6} 
  \dots  & 0 & 0 & 0 & 0 & \dot{a}_{\integers,t}(1,2) & \dot{b}_{\integers}(2,t) &  \dots  \\
  & \iddots   & \vdots & \vdots &  \vdots & \vdots   & \ddots  &  \\
  \end{array}
  \right)
\end{equation*}
The radial Lax operator on $G$ is directly computed via $ \Plax_{rad}(t) =  \pp ^{\ast}    \PlaxZ (t)  \pp $, see Appendix \ref{app:liftingLax}. We obtain the following matrix elements for the commutator $[\Plax_{rad}, \Lhamilton] $.   
Note that we restrict ourselves to the submatrix  corresponding to the vertices $\{-1,0w_1,0w_2,1\}$.

\begin{equation*}
  \renewcommand{\arraystretch}{1.7}
  \begin{array}{c c c c c}
   \cline{1-4} 
\multicolumn{1}{|c}{2a_{\integers,t}(-1,0)^2-2a_{\integers,t}(-2,-1)^2}  & \frac{a_{\integers,t}(-1,0)}{2}(b_{\integers}(0,t)-b_{\integers}(-1,t)) & \frac{a_{\integers,t}(-1,0)}{2}(b_{\integers}(0,t)-b_{\integers}(-1,t))  & \multicolumn{1}{c|}{0 \  \  \text{    }}  \\
  \multicolumn{1}{|c}{a_{\integers,t}(-1,0)(b_{\integers}(0,t)-b_{\integers}(-1,t))}  &   a_{\integers,t}(0,1)^2-a_{\integers,t}(-1,0)^2 &  a_{\integers,t}(0,1)^2-a_{\integers,t}(-1,0)^2 &  \multicolumn{1}{c|}{\ast  \  \  \text{    }}   \\
  \multicolumn{1}{|c}{a_{\integers,t}(-1,0)(b_{\integers}(0,t)-b_{\integers}(-1,t))} & a_{\integers,t}(0,1)^2-a_{\integers,t}(-1,0)^2 &   a_{\integers,t}(0,1)^2-a_{\integers,t}(-1,0)^2 &  \multicolumn{1}{c|}{\ast  \  \  \text{    } }   \\
   \multicolumn{1}{|c}{0}    & \frac{a_{\integers,t}(0,1)}{2}(b_{\integers}(1,t)-b_{\integers}(0,t)) & \frac{a_{\integers,t}(0,1)}{2}(b_{\integers,t}(1,t)-b_{\integers,t}(0,t)) & \multicolumn{1}{c|}{ \ast \  \  \text{    }    } \\ 
   \cline{1-4} 
  \\
  \end{array}
\end{equation*}
Comparing matrix elements from $\frac{d}{dt}\Lhamilton \proj  $ with  $[\Plax_{rad}, \Lhamilton] $, we obtain a system of equations that indeed reproduce the equations of motion (\ref{eq: TransformedequationsOfMotiona1}), (\ref{eq: TransformedequationsOfMotiona2}) and (\ref{eq: TransformedequationsOfMotiona3})  with radial initial data, i.e.,
\begin{equation*}
 a_t(-1,0w_1)|_{t=0} = a_t(-1,0w_2)|_{t=0}, \quad a_t(0w_1,1)|_{t=0} = a_t(0w_2,1)|_{t=0}, \quad b(0w_1,t)|_{t=0} = b(0w_2,t)|_{t=0}.
\end{equation*}
This is a consequence of Theorem \ref{thm:liftToradailSolution} and can be directly verified when simplifying the notation in equations  (\ref{eq: TransformedequationsOfMotiona1}), (\ref{eq: TransformedequationsOfMotiona2}) and (\ref{eq: TransformedequationsOfMotiona3}). To this end, we set
\begin{itemize}
	\item $a_{\integers,t}(-1,0) := a_t(-1,0w_1) = a_t(-1,0w_2)$,
	\item $ a_{\integers,t}(0,1) := a_t(0w_1,1) = a_t(0w_2,1)$,
	\item $b_{\integers}(0,t) := b(0w_1,t) = b(0w_2,t)$.
	\item $a_{\integers,t}(n,n+1) := a_t(n,n+1)$ for $n \in \integers \backslash \{-1,0\}$
	\item $b_{\integers}(n,t) := b(n,t)$ for $n \in \integers \backslash \{0\}$
\end{itemize}

We turn   to the question: \textit{Can we find a skew-adjoint non-zero operator $ \Plax_{\bot}(t)$ which   reflects the adjacency relation of the graph  such that $\frac{d}{dt} \Lhamilton (t)(\Id- \proj  ) = [\Plax_{\bot}(t) ,\  \Lhamilton(t) ]$?}  If the answer is in the affirmative, then the isospectral property will imply that the $\Lhamilton(t)$-spectrum is $t$-independent, i.e., $\sigma(\Lhamilton(t))=\sigma(\Lhamilton(0))$. But this is  true only in the trivial case given in the following theorem. 
\begin{theorem}
\label{noLaxTheorem}Assume that, in  Example~\ref{ex:ContinuationOfMainExample},   $\Lhamilton(t)$ is the operator given in \eqref{eq:equForliftingThm} or, equivalently,  (\ref{exampleHamilton}). This operator acts on $G$ and lifts   a Jacobi matrix $\LhamiltonZ(t)$, given in \eqref{eq:infiniteJacobiMatrix}, in the sense of Definition \ref{def:LiftOperator}.

Then any isospectral dynamic of  $\Lhamilton(t)$ will  impose the constraint $\frac{d}{dt}b(0,t) =0$, which implies 
$$ [\PlaxG(t),\Lhamilton(t)]  =0 $$ 
and the  trivial dependence on time $t$: $\PlaxG(t)=\PlaxG(0)$ and $\Lhamilton(t) =\Lhamilton(0)$.
\end{theorem}
\begin{proof}
We first investigate the restriction of $\Lhamilton(t)$ to  the subspace $Ker  \pp $ because, due to Corollary \ref{coroForComputLift}, $\Lhamilton(t)$ satisfies the Spectral Separation Assumption, i.e., Assumption \ref{separationAnsatz}. For a vertex $x \in V(G)$, the vector $e_{x}$ corresponds to
\begin{equation*}
e_{x} =
  \begin{cases}
    1       & \quad \text{on vertex } x\\
    0  & \quad \text{on } V(G) \backslash \{x\} .
  \end{cases}
\end{equation*}
Now, it can be shown that $Ker  \pp  = span\{ \  e_{0w_1} - e_{0w_2}  \  \} $. We compute 
\\
\text{ }
\begin{equation*}
 \renewcommand{\arraystretch}{1.5}
  \Lhamilton(t) (\Id _G - \proj  ) =
  \left(
  \begin{array}{c c c  c c c c  c c c }
     \ddots      &  \vdots   &  \vdots & \vdots &  \vdots  &  \vdots  & \vdots & \iddots  \\
     \dots        & 0          & 0         & 0        & 0           & 0          & 0        &  \dots \\
     \cline{3-6} 
     \dots       &  \multicolumn{1}{c|}{0} & 0 & 0 & 0 & \multicolumn{1}{c|}{0} & 0 &  \dots \\
 \dots    & \multicolumn{1}{c|}{0} & 0 &  \frac{b(0,t)}{2} &- \frac{b(0,t)}{2} &  \multicolumn{1}{c|}{0}  & 0 & \dots \\
 \dots    & \multicolumn{1}{c|}{0} &  0 & - \frac{b(0,t)}{2} &  \frac{b(0,t)}{2} &  \multicolumn{1}{c|}{0}  &  0 & \dots \\
  \dots & \multicolumn{1}{c|}{0}   & 0 & 0 & 0 & \multicolumn{1}{c|}{ 0}  & 0 & \dots \\ 
   \cline{3-6} 
  \dots &0 & 0 & 0 & 0 & 0 & 0 &    \dots \\
  \iddots  & \vdots  & \vdots & \vdots & \vdots &  \vdots  &  \vdots  & \ddots\\
  \end{array}
  \right)
\end{equation*}
\text{ }
\\
\\
and obtain for the restriction $\Lhamilton_{\bot}(t) = \Lhamilton(t)|_{Ker  \pp  }$,   $\Lhamilton_{\bot}(t): Ker  \pp  \to Ker  \pp , \ v \longmapsto b(0,t) v$. Then we have  spectrum of $\Lhamilton_{\bot}(t)$
\begin{equation}
\sigma(\Lhamilton_{\bot}(t) ) =\{b(0,t) \}.
\end{equation}
Hence, achieving isospectrality would impose the constraint $\frac{d}{dt}b(0,t) =0$. 
This constraint means
that the momenta of the particles at the vertices  $0w_1$ and $0w_2$ are time-independent. Moreover, the interactions between the particles at $0w_1$ and $0w_2$ and the particles at the vertices $-1$ and $1$ are equal at all times,
\begin{equation*}
a_{\integers,t}(-1,0)=a_{\integers,t}(0,1).
\end{equation*}
This implies, by induction in $n$, that $a_{\integers,t}(n,n+1)$ are equal for all $t$ and $n$ using \eqref{eq: TransformedequationsOfMotion1D} or the standard transfer matrix formalism \cite{AkkermansDunneLevy,TeschlbookJacobi,teschlmathematical}. 
Note   that the constraint $\frac{d}{dt}b(0,t) =0$ implies $\frac{d}{dt} \Lhamilton(t)(\Id  - \proj  )=0$ and hence $\frac{d}{dt} \Lhamilton (t)  = [\Plax_{rad}(t) ,\  \Lhamilton(t) ]$, which is related to Example~\ref{ex:lastExample}. \end{proof}\end{example}

\begin{remark}\label{rem:final}
We prove Theorem~\ref{noLaxTheorem} only in the situation of  Example~\ref{ex:ContinuationOfMainExample} 
in order to simplify notation 
although the same result can be expected for all examples considered above in Figures~\ref{fig:Hambly-Kumagai}, \ref{fig:simple} and \ref{fig:simplestZgradedGraph}.  
\end{remark}

\begin{example} 
	\label{ex:lastExample}Proposition~\ref{prop:RadialLax} uses the simple fact that $\pp$ is an isomorphism between $  \ell^2_{rad}(G)$ and   $\ell^2 ( \integers)$.
Therefore, if $\LhamiltonZ(t)$ is an isospectral family on  $\ell^2(\integers)$, then  we always can trivially generate   an isospectral family on $\ell^2(G)$ by \begin{equation}\label{e-unison}
\Lhamilton(t)=\Lhamilton_{rad}(t)=\pp ^{\ast}    \LhamiltonZ (t)  \pp  
\end{equation}with\begin{equation}
\Plax_{\bot}(t)=\Lhamilton_{\bot}(t)=0 ,
\qquad 
\PlaxG (t)=\Plax_{rad}(t)=\pp ^{\ast}    \PlaxZ (t)  \pp.
\end{equation}

In interesting $\integers$-graded examples, such as Figures~\ref{fig:Hambly-Kumagai}, \ref{fig:simple} and \ref{fig:simplestZgradedGraph}, $\Lhamilton_{rad}(t)=\pp ^{\ast}    \LhamiltonZ (t)  \pp$ 
does not   reflect  the adjacency relation of the graph. However the same  $\Lhamilton_{rad}(t)$ in Example~\ref{ex:ContinuationOfMainExample} reflects the   adjacency relation of the different graph, see  
  Figure~\ref{fig:rem}.
  
 \begin{figure}[htb]
 	\centering
\includegraphics{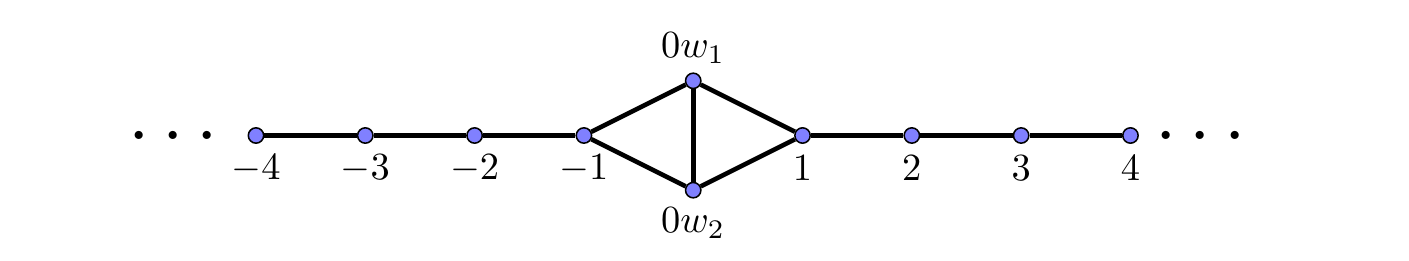}
\caption{Graph $G$ in Example~\ref{ex:lastExample}, which is not a proper $\integers$-graded graph. }
\label{fig:rem}
\end{figure} 
On the graph $G$ in  Figure~\ref{fig:rem} waves corresponding to \eqref{e-unison}  can move in the horizontal, or radial, direction with vertices $0w_1$ and $0w_2$ vibrating in unison. However, note that the graph in Figure~\ref{fig:rem} is not a proper $\integers$-graded graph because the vertical edge does not have a natural direction.
\end{example} 

\section{A relation to the non-Abelian Toda lattice} \label{non-Abelian}
The non-Abelian Toda lattice was introduced by Polyakov 
as a discretization of the principal chiral field equation. It was solved by means of the inverse scattering method in \cite{BMRL} (see also \cite{EGR97,Gekhtman98,KricheverZabrodin95, Widom98} for more details and various aspects of the theory). Roughly speaking, the non-Abelian Toda lattice has the form similar to \eqref{eq: TransformedequationsOfMotion1D}  where $a$'s and $b$'s are matrices. {We mention this in our article because, for instance,  the matrix $\hamilton$ corresponding to the middle graph in Figure \ref{fig:simple}, which is a 2-periodic graph, has a natural pentadiagonal structure. 
This matrix can be considered as a block tridiagonal matrix, whose blocks are $2\times 2$ matrices. Moreover, any such periodic graph will lead to a banded matrix (e.g. the bottom graph in Figure \ref{fig:simple} is 4-periodic and the corresponding $\hamilton$ is again a block tridiagonal matrix whose blocks are $4\times 4$ matrices.)
However, such structures do not necessarily fit the assumptions of the non-Abelian Toda lattice since some of the non-diagonal matrices could be singular, which can be seen in the case of the pentadiagonal matrix built on the middle graph in Figure \ref{fig:simple}.  {Besides, in our consideration the structure of the corresponding blocks is imposed by the geometry of the corresponding graph, while for the non-Abelian Toda lattice the matrix structure is determined by the symmetry properties of the matrix-valued dynamical functions themselves,
not necessarily related to a fixed graph structure.}   {However, once the solution to the non-Abelian Toda lattice is found, it defines certain properties of the graph structure and, in particular, it could lead to periodic graphs.} } {This raises the interesting possibility of non-trivial relations between different sources of symmetry, and such graph structures and the possibility to generalize our lifting procedure to fit into the case of non-Abelian Toda lattice will be the subject of further study}. 

\section{Conclusions}
{In this article we demonstrate that nonlinear one-dimensional Hamiltonian 
dynamics can be lifted to $\integers$-graded graphs, provided certain  conditions hold, most important of which is the Spectral Separation Assumption,  
Assumption \ref{separationAnsatz}. The most important {\em sufficient} assumption for this is {the symmetry assumption for the measure,} Assumption~\ref{necessaryAssumptionforRadialSol}, which is used in Theorem~\ref{thm:LiftingTheorem}. This lifting of Hamiltonian 
dynamics means that for a large class of  $\integers$-graded graphs there are soliton-like solutions for non-linear wave equations. On the other hand, we show that the commutator relation for non-trivial Lax pairs can not be lifted using this procedure, see Theorem~\ref{noLaxTheorem} and also  \cite{BerezanskyShmois}. 

We expect that our methods are applicable to a wide variety of quantum graphs and limiting fractals, using methods of 
\cite{hanoi,AR18,AR21,MR3154573,Berkolaiko,BKbook,ST19}.  
{In particular, our results are compatible with
recent related results for Dirac operators on quantum graphs \cite{SabirovBMK}, 
if we specialize our combinatorial graphs to quantum graphs. However we are building a combinatorially more general framework, which can support much larger and more complicated classes of graphs by identifying general conditions for propagation of non-linear waves. For example, the structures in \cite[Figures 5 and 6]{SabirovBMK}
fit within our framework, and \cite[Figure 2]{SabirovBMK} shows the behavior that we expect to observe in our more general framework. 
It can be said that the quantum graphs in \cite{SabirovBMK,AR18,AR21} are $\mathbb R$-graded, while ours are $\integers$-graded, reflecting the usual difference between discrete and continuous settings.}

Thus, we can highlight two aspects of our work: non-linear waves can be lifted to $\integers$-graded graphs, and possibly $\mathbb R$-graded graphs and fractals, but general integrability can not be lifted in this framework. We hope that our result will show how one can manipulate solitons on complicated networks, following \cite{Christodoulides}, and study more general time dependent situations, see \cite{transparent,KayMoses}.

}

\subsection*{Data Availability Statement} Data sharing not applicable -- no new data generated.
\subsection*{Acknowledgments}

{This research was supported in part by 
	the University of Connecticut Research Excellence Program, 
	by DOE grant DE-SC0010339 and by NSF DMS grants 1613025 and 2008844.}
The authors    thank    the    anonymous    reviewers    for    their    careful    reading    of    our    manuscript    and    their    many    insightful    comments    and    suggestions.


\begin{thebibliography}{BCD{\etalchar{+}}08b}
	
	\bibitem[AACC91]{ablowitz1991solitons}
	Mark~J Ablowitz, MA~Ablowitz, PA~Clarkson, and Peter~A Clarkson.
	\newblock {\em Solitons, nonlinear evolution equations and inverse scattering},
	volume 149.
	\newblock Cambridge university press, 1991.
	
	\bibitem[ABD{\etalchar{+}}12]{ABDTV12}
	Eric Akkermans, Olivier Benichou, Gerald~V Dunne, Alexander Teplyaev, and
	Raphael Voituriez.
	\newblock Spatial log-periodic oscillations of first-passage observables in
	fractals.
	\newblock {\em Physical Review E}, 86(6):061125, 2012.
	
	\bibitem[ACD{\etalchar{+}}19]{ACDRT}
	Eric Akkermans, Joe~P Chen, Gerald Dunne, Luke~G Rogers, and Alexander
	Teplyaev.
	\newblock Fractal {AC} circuits and propagating waves on fractals.
	\newblock {\em to appear in the 6th Cornell Fractals Conference Proceedings,
		arXiv:1507.05682}, 2019.
	
	\bibitem[ADL13]{AkkermansDunneLevy}
	Eric Akkermans, Gerald~V. Dunne, and Eli Levy.
	\newblock Wave propagation in one-dimension: Methods and applications to
	complex and fractal structures.
	\newblock In {\em Optics of aperiodic structures: fundamentals and device
		applications, edited by Luca Dal Negro}, pages 407--450
	\href{https://arxiv.org/abs/1210.7409}{arXiv:1210.7409}
	\href{https://doi.org/10.1201/b15653}{doi:10.1201/b15653}. Taylor \& Francis
	CRC Press, 2013.
	
	\bibitem[ADT09]{ADT09}
	Eric Akkermans, Gerald~V Dunne, and Alexander Teplyaev.
	\newblock Physical consequences of complex dimensions of fractals.
	\newblock {\em EPL (Europhysics Letters)}, 88(4):40007, 2009.
	
	\bibitem[ADT10]{ADT10}
	Eric Akkermans, Gerald~V Dunne, and Alexander Teplyaev.
	\newblock Thermodynamics of photons on fractals.
	\newblock {\em Physical review letters}, 105(23):230407, 2010.
	
	\bibitem[Agr13]{agrawal2013nonlinear}
	Govind~P Agrawal.
	\newblock {\em Nonlinear Fiber Optics: Formerly Quantum Electronics}.
	\newblock Academic press, 2013.
	
	\bibitem[Akk13]{Akk}
	Eric Akkermans.
	\newblock Statistical mechanics and quantum fields on fractals.
	\newblock In {\em Fractal geometry and dynamical systems in pure and applied
		mathematics. {II}. {F}ractals in applied mathematics}, volume 601 of {\em
		Contemp. Math.}, pages 1--21. Amer. Math. Soc., Providence, RI, 2013.
	
	\bibitem[AMR88]{amr}
	R.~Abraham, J.~E. Marsden, and T.~Ratiu.
	\newblock {\em Manifolds, tensor analysis, and applications}, volume~75 of {\em
		Applied Mathematical Sciences}.
	\newblock Springer-Verlag, New York, second edition, 1988.
	
	\bibitem[AR18]{AR18}
	Patricia Alonso~Ruiz.
	\newblock Explicit formulas for heat kernels on diamond fractals.
	\newblock {\em Comm. Math. Phys.}, 364(3):1305--1326, 2018.
	
	\bibitem[AR21]{AR21}
	Patricia Alonso~Ruiz.
	\newblock Heat kernel analysis on diamond fractals.
	\newblock {\em Stochastic Processes and their Applications}, 131:51--72, 2021.
	
	\bibitem[ARKT16]{hanoi}
	Patricia Alonso-Ruiz, Daniel~J. Kelleher, and Alexander Teplyaev.
	\newblock Energy and {L}aplacian on {H}anoi-type fractal quantum graphs.
	\newblock {\em J. Phys. A}, 49(16):165206, 36, 2016.
	
	\bibitem[BBJL17]{Berkolaiko}
	Ram Band, Gregory Berkolaiko, Christopher~H. Joyner, and Wen Liu.
	\newblock Quotients of finite-dimensional operators by symmetry
	representations.
	\newblock {\em arXiv:1711.00918}, 2017.
	
	\bibitem[BCD{\etalchar{+}}08a]{v1}
	N.~Bajorin, T.~Chen, A.~Dagan, C.~Emmons, M.~Hussein, M.~Khalil, P.~Mody,
	B.~Steinhurst, and A.~Teplyaev.
	\newblock Vibration modes of {$3n$}-gaskets and other fractals.
	\newblock {\em J. Phys. A}, 41(1):015101, 21, 2008.
	
	\bibitem[BCD{\etalchar{+}}08b]{v2}
	N.~Bajorin, T.~Chen, A.~Dagan, C.~Emmons, M.~Hussein, M.~Khalil, P.~Mody,
	B.~Steinhurst, and A.~Teplyaev.
	\newblock Vibration spectra of finitely ramified, symmetric fractals.
	\newblock {\em Fractals}, 16(3):243--258, 2008.
	
	\bibitem[BK13a]{BKbook}
	Gregory Berkolaiko and Peter Kuchment.
	\newblock {\em Introduction to quantum graphs}, volume 186 of {\em Mathematical
		Surveys and Monographs}.
	\newblock American Mathematical Society, Providence, RI, 2013.
	
	\bibitem[BK13b]{MR3154573}
	Jonathan Breuer and Matthias Keller.
	\newblock Spectral analysis of certain spherically homogeneous graphs.
	\newblock {\em Oper. Matrices}, 7(4):825--847, 2013.
	
	\bibitem[BMRL80]{BMRL}
	M.~Bruschi, S.~V. Manakov, O.~Ragnisco, and D.~Levi.
	\newblock The nonabelian {T}oda lattice-discrete analogue of the matrix
	{S}chr\"{o}dinger spectral problem.
	\newblock {\em J. Math. Phys.}, 21(12):2749--2753, 1980.
	
	\bibitem[BO16]{borodin_olshanski_2016}
	Alexei Borodin and Grigori Olshanski.
	\newblock {\em Representations of the Infinite Symmetric Group}.
	\newblock Cambridge Studies in Advanced Mathematics. Cambridge University
	Press, 2016.
	
	\bibitem[BS94]{BerezanskyShmois}
	Yurij Berezansky and Michael Shmoish.
	\newblock Nonisospectral flows on semi-infinite {J}acobi matrices.
	\newblock {\em J. Nonlinear Math. Phys.}, 1(2):116--146, 1994.
	
	\bibitem[CE01]{Christodoulides}
	Demetrios~N. Christodoulides and Eugenia~D. Eugenieva.
	\newblock Blocking and routing discrete solitons in two-dimensional networks of
	nonlinear waveguide arrays.
	\newblock {\em Phys. Rev. Lett.}, 87, 2001.
	\newblock DOI: 10.1103/PhysRevLett.87.233901.
	
	\bibitem[DDMT20]{2019arXiv190908668D}
	Maxim Derevyagin, Gerald~V. Dunne, Gamal Mograby, and Alexander Teplyaev.
	\newblock Perfect quantum state transfer on diamond fractal graphs.
	\newblock {\em Quantum Information Processing}, 19(9):328, September 2020.
	
	\bibitem[DT13]{transparent}
	Gerald~V Dunne and Michael Thies.
	\newblock Transparent {D}irac potentials in one dimension: The time-dependent
	case.
	\newblock {\em Physical Review A}, 88(6):062115, 2013.
	
	\bibitem[Dun12]{Dunne12}
	Gerald~V. Dunne.
	\newblock Heat kernels and zeta functions on fractals.
	\newblock {\em J. Phys. A}, 45(37):374016, 22, 2012.
	
	\bibitem[EGR97]{EGR97}
	Pavel Etingof, Israel Gelfand, and Vladimir Retakh.
	\newblock Factorization of differential operators, quasideterminants, and
	nonabelian {T}oda field equations.
	\newblock {\em Math. Res. Lett.}, 4(2-3):413--425, 1997.
	
	\bibitem[Fla74a]{10.1143/PTP.51.703}
	H.~Flaschka.
	\newblock {On the {T}oda Lattice. II: Inverse-Scattering Solution}.
	\newblock {\em Progress of Theoretical Physics}, 51(3):703--716, 03 1974.
	
	\bibitem[Fla74b]{PhysRevB.9.1924}
	H.~Flaschka.
	\newblock The {T}oda lattice. ii. existence of integrals.
	\newblock {\em Phys. Rev. B}, 9:1924--1925, Feb 1974.
	
	\bibitem[Fom94]{Fomin}
	Sergey Fomin.
	\newblock Duality of graded graphs.
	\newblock {\em J. Algebraic Combin.}, 3(4):357--404, 1994.
	
	\bibitem[FT87]{FaddeevTakhtajan}
	L.~D. Faddeev and L.~A. Takhtajan.
	\newblock {\em Hamiltonian methods in the theory of solitons}.
	\newblock Springer Series in Soviet Mathematics. Springer-Verlag, Berlin, 1987.
	\newblock Translated from the Russian by A. G. Reyman [A. G. Re\u{\i}man].
	
	\bibitem[Gek98]{Gekhtman98}
	M.~Gekhtman.
	\newblock Hamiltonian structure of non-abelian {T}oda lattice.
	\newblock {\em Lett. Math. Phys.}, 46(3):189--205, 1998.
	
	\bibitem[H\'74]{PhysRevB.9.1921}
	M.~H\'enon.
	\newblock Integrals of the {T}oda lattice.
	\newblock {\em Phys. Rev. B}, 9:1921--1923, Feb 1974.
	
	\bibitem[HM20]{HM19}
	Michael Hinz and Melissa Meinert.
	\newblock On the viscous {B}urgers equation on metric graphs and fractals.
	\newblock {\em J. Fractal Geom.}, 7(2):137--182, 2020.
	
	\bibitem[HO07]{MR2316893}
	Akihito Hora and Nobuaki Obata.
	\newblock {\em Quantum probability and spectral analysis of graphs}.
	\newblock Theoretical and Mathematical Physics. Springer, Berlin, 2007.
	\newblock With a foreword by Luigi Accardi.
	
	\bibitem[KM56]{KayMoses}
	I.~Kay and H.~E. Moses.
	\newblock Reflectionless transmission through dielectrics and scattering
	potentials.
	\newblock {\em Journal of Applied Physics}, 27(12):1503--1508, 1956.
	
	\bibitem[KT04]{KT}
	Bernhard Kr\"{o}n and Elmar Teufl.
	\newblock Asymptotics of the transition probabilities of the simple random walk
	on self-similar graphs.
	\newblock {\em Trans. Amer. Math. Soc.}, 356(1):393--414, 2004.
	
	\bibitem[KZ95]{KricheverZabrodin95}
	Igor~Moiseevich Krichever and A.~Zabrodin.
	\newblock Spin generalization of the {R}uijsenaars-{S}chneider model, the
	nonabelian two-dimensionalized {T}oda lattice, and representations of the
	{S}klyanin algebra.
	\newblock {\em Uspekhi Mat. Nauk}, 50(6(306)):3--56, 1995.
	
	\bibitem[Lax68]{doi:10.1002/cpa.3160210503}
	Peter~D. Lax.
	\newblock Integrals of nonlinear equations of evolution and solitary waves.
	\newblock {\em Communications on Pure and Applied Mathematics}, 21(5):467--490,
	1968.
	
	\bibitem[LP01]{LP}
	Urs Lang and Conrad Plaut.
	\newblock Bilipschitz embeddings of metric spaces into space forms.
	\newblock {\em Geom. Dedicata}, 87(1-3):285--307, 2001.
	
	\bibitem[Man75]{manakov}
	S.~V. Manakov.
	\newblock Complete integrability and stochastization of discrete dynamical
	systems.
	\newblock {\em Soviet Journal of Experimental and Theoretical Physics},
	40(2):269--274, 1975.
	
	\bibitem[MDDT20]{mograby2020spectra}
	Gamal Mograby, Maxim Derevyagin, Gerald~V Dunne, and Alexander Teplyaev.
	\newblock Spectra of perfect state transfer hamiltonians on fractal-like
	graphs.
	\newblock {\em Journal of Physics A: Mathematical and Theoretical,
		(arXiv:2003.11190)}, 2020.
	
	\bibitem[Mos75]{Moser1975}
	J\"{u}rgen Moser.
	\newblock {\em Finitely many mass points on the line under the influence of an
		exponential potential--an integrable system}, pages 467--497. Lecture Notes
	in Phys., Vol. 38.
	\newblock Springer, Berlin, 1975.
	
	\bibitem[MT95]{MT}
	Leonid Malozemov and Alexander Teplyaev.
	\newblock Pure point spectrum of the {L}aplacians on fractal graphs.
	\newblock {\em J. Funct. Anal.}, 129(2):390--405, 1995.
	
	\bibitem[MT03]{MT2}
	Leonid Malozemov and Alexander Teplyaev.
	\newblock Self-similarity, operators and dynamics.
	\newblock {\em Math. Phys. Anal. Geom.}, 6(3):201--218, 2003.
	
	\bibitem[NMPZ84]{NMPZ}
	S.~Novikov, S.~V. Manakov, L.~P. Pitaevski\u{\i}, and V.~E. Zakharov.
	\newblock {\em Theory of solitons}.
	\newblock Contemporary Soviet Mathematics. Consultants Bureau [Plenum], New
	York, 1984.
	\newblock The inverse scattering method, Translated from the Russian.
	
	\bibitem[SBK18]{SabirovBMK}
	K.~K. Sabirov, D.~B. Babajanov, and D.~U. Matrasulovand P.~G. Kevrekidis.
	\newblock Dynamics of {D}irac solitons in networks.
	\newblock {\em J. Phys. A: Math. Theor.}, 51:435203, 2018.
	\newblock doi.org/10.1088/1751-8121/aadfb0.
	
	\bibitem[ST20]{ST19}
	Benjamin Steinhurst and Alexander Teplyaev.
	\newblock Spectral analysis and dirichlet forms on {B}arlow-{E}vans fractals.
	\newblock {\em To appear in the Journal of Spectral Theory, arXiv:1204.5207},
	2020.
	
	\bibitem[Sta88]{MR941434}
	Richard~P. Stanley.
	\newblock Differential posets.
	\newblock {\em J. Amer. Math. Soc.}, 1(4):919--961, 1988.
	
	\bibitem[Sta12]{Stanley}
	Richard~P. Stanley.
	\newblock {\em Enumerative combinatorics. {V}olume 1}, volume~49 of {\em
		Cambridge Studies in Advanced Mathematics}.
	\newblock Cambridge University Press, Cambridge, second edition, 2012.
	
	\bibitem[SZ72]{shabat1972exact}
	A~Shabat and V~Zakharov.
	\newblock Exact theory of two-dimensional self-focusing and one-dimensional
	self-modulation of waves in nonlinear media.
	\newblock {\em Soviet physics JETP}, 34(1):62, 1972.
	
	\bibitem[Tes00]{TeschlbookJacobi}
	Gerald Teschl.
	\newblock {\em Jacobi operators and completely integrable nonlinear lattices},
	volume~72 of {\em Mathematical Surveys and Monographs}.
	\newblock American Mathematical Society, Providence, RI, 2000.
	
	\bibitem[Tes14]{teschlmathematical}
	Gerald Teschl.
	\newblock {\em Mathematical methods in quantum mechanics. With applications to
		Schr\"{o}dinger operators}, volume 157 of {\em Graduate Studies in
		Mathematics}.
	\newblock American Mathematical Society, Providence, RI, second edition, 2014.
	
	\bibitem[Tod67a]{doi:10.1143/JPSJ.22.431}
	Morikazu Toda.
	\newblock Vibration of a chain with nonlinear interaction.
	\newblock {\em Journal of the Physical Society of Japan}, 22(2):431--436, 1967.
	
	\bibitem[Tod67b]{doi:10.1143/JPSJ.23.501}
	Morikazu Toda.
	\newblock Wave propagation in anharmonic lattices.
	\newblock {\em Journal of the Physical Society of Japan}, 23(3):501--506, 1967.
	
	\bibitem[Tod81]{nla.cat-vn2968918}
	Morikazu Toda.
	\newblock {\em Theory of nonlinear lattices / Morikazu Toda}.
	\newblock Springer-Verlag Berlin ; New York, 1981.
	
	\bibitem[Wid98]{Widom98}
	Harold Widom.
	\newblock An integral operator solution to the matrix {T}oda equations.
	\newblock {\em J. Integral Equations Appl.}, 10(3):363--372, 1998.
	
\end{thebibliography}

\newcommand{\etalchar}[1]{$^{#1}$}
\def\cprime{$'$}

\newpage
\appendix

\section{Matrix representations}

\subsection{The averaging operator $ \pp $, its adjoint $ \pp ^{\ast}$ and the projector $\proj  $ in Example \ref{ex:MainExample}}
The averaging operator $ \pp $ (see Definition \ref{def:radProjAvg}), its adjoint $ \pp ^{\ast}$ (see Proposition \ref{propForaveraging}), orthogonal projection $\proj  $ onto the subspace of radial functions $\ell^2_{rad}(G)$ (see Definition \ref{def:radProjAvg}) for Example \ref{ex:MainExample}:
\begin{equation*}
  \renewcommand{\arraystretch}{1.5}
   \pp =
  \left(
  \begin{array}{c c c  c c c c  c c c }
     &  \ddots   &  \vdots  &  \vdots &  \vdots  &  \vdots  & \vdots & \iddots  \\
   \ddots  & 1 & 0  & 0 & 0 & 0 & 0 &\dots\\
  \dots & 0 & 1 & 0 & 0 & 0 & 0 & \dots \\
  \dots &0 & 0 & \frac{1}{2} & \frac{1}{2} & 0 & 0 &  \dots  \\
    \dots &0 & 0 & 0 & 0 & 1 & 0 & \dots   \\
        \iddots  &0 & 0 & 0 & 0 & 0 & 1 & \ddots  \\
  & \iddots  & \vdots & \vdots & \vdots  &  \vdots  & \ddots\\
  \end{array}
  \right),
  \quad \quad
     \pp ^{\ast}=
  \left(
  \begin{array}{c c c  c c c c  c c c }
     &  \ddots   &   \vdots  &  \vdots  & \vdots & \iddots  \\
   \ddots  & 1 & 0  & 0   & 0 &\dots\\
  \dots & 0 & 1 & 0 &  0 & \dots \\
  \dots &0 & 0 & 1 &  0 &  \dots  \\
    \dots &0 & 0 &  1 & 0 & \dots   \\
        \iddots  &0 & 0 & 0 &  1 & \ddots  \\
  & \iddots  & \vdots  &  \vdots  & \ddots\\
  \end{array}
  \right)
\end{equation*}
\begin{equation*}
  \renewcommand{\arraystretch}{1.5}
  \proj  =
  \left(
  \begin{array}{c c c  c c c c  c c c }
     &  \ddots   &  \vdots  &  \vdots &  \vdots  &  \vdots  & \vdots & \iddots  \\
   \ddots  & 1 & 0  & 0 & 0 & 0 & 0 &\dots\\
  \dots & 0 & 1 & 0 & 0 & 0 & 0 & \dots \\
  \dots &0 & 0 & \frac{1}{2} & \frac{1}{2} & 0 & 0 &  \dots  \\
   \dots &0 & 0 & \frac{1}{2} & \frac{1}{2} & 0 & 0 &  \dots  \\
    \dots &0 & 0 & 0 & 0 & 1 & 0 & \dots   \\
        \iddots  &0 & 0 & 0 & 0 & 0 & 1 & \ddots  \\
  & \iddots  & \vdots & \vdots & \vdots  &  \vdots  & \ddots\\
  \end{array}
  \right).
\end{equation*}

\subsection{The radial Lax operator $\Plax_{rad}(t)$  in Example \ref{ex:MainExample}} 
\label{app:liftingLax}
Recall  $\PlaxZ (t) := [\jacobi(t)]_{+} - [\jacobi(t)]_{-}$, where  $[\jacobi(t)]_{+}$  and $ [\jacobi(t)]_{-}$ define the upper and lower triangular parts of $\jacobi(t)$ respectively. The corresponding radial Lax operator on $G$ is is directly computed via $ \Plax_{rad}(t) =  \pp ^{\ast}    \PlaxZ (t)  \pp $. We obtain
\begin{equation*}
  \renewcommand{\arraystretch}{1.5}
  \Plax_{rad}(t)=
  \left(
  \begin{array}{c c c  c c c c  c c c }
       \ddots     &  \vdots & \vdots &  \vdots  &  \vdots  & \vdots &  \vdots & \iddots  \\
  \dots  & 0 & a_{\integers, t}(-2,-1) & 0 & 0 & 0 & 0  &  \dots \\
   \cline{3-6} 
\dots    &  \multicolumn{1}{c|}{-a_{\integers, t}(-2,-1)} & 0 & \frac{a_{\integers, t}(-1,0)}{2} & \frac{a_{\integers, t}(-1,0)}{2} & \multicolumn{1}{c|}{0} & 0 &  \dots \\
 \dots     & \multicolumn{1}{c|}{0} & - a_{\integers, t}(-1,0) & 0 & 0 &  \multicolumn{1}{c|}{a_{\integers, t}(0,1)}  &  0 & \dots \\
 \dots     & \multicolumn{1}{c|}{0} & - a_{\integers, t}(-1,0) & 0 & 0 &  \multicolumn{1}{c|}{a_{\integers, t}(0,1)}  &  0 & \dots \\
  \dots  & \multicolumn{1}{c|}{0}   & 0 & -\frac{a_{\integers, t}(0,1)}{2} & -\frac{a_{\integers, t}(0,1)}{2} & \multicolumn{1}{c|}{0} &  a_{\integers, t}(1,2) & \dots \\ 
   \cline{3-6} 
  \dots  & 0 & 0 & 0 & 0 & - a_{\integers, t}(1,2) & 0 & \dots  \\
   \iddots  & \vdots & \vdots & \vdots &  \vdots &   \vdots &  \ddots\\
  \end{array}
  \right)
\end{equation*}
The radial Lax operator $\Plax_{rad}(t)$ is almost identical with $\PlaxZ (t)$, except at the rows and columns corresponding to the vertices $\{-1,0w_1,0w_2,1\}$. The corresponding submatrix is indicated inside a rectangle in the middle of the matrix.

\end{document}